\documentclass[11pt]{article}
\usepackage[utf8]{inputenc}
\usepackage[margin=1in]{geometry}
\usepackage{mathtools, amsmath, amsthm, thm-restate, graphicx, enumitem, amsfonts, amssymb, verbatim, bbm, braket, xcolor, float, complexity, comment, tikz, tikz-3dplot, url, mdframed, needspace, titlesec, algorithm}
\usepackage[noend]{algpseudocode}
\usepackage[linktocpage=true]{hyperref}
\usepackage[capitalise]{cleveref}
\usepackage[most]{tcolorbox}

\algnewcommand\algorithmicinput{\textbf{Input:}}
\algnewcommand\Input{\item[\algorithmicinput]}

\newmdenv[linecolor=black, linewidth=1pt]{factbox}

\newcommand{\knote}[1]{{\color{brown} (Kunal: #1)}}
\newcommand{\jnote}[1]{{\color{teal} (James: #1)}}
\newcommand{\anote}[1]{{\color{orange} (Anuj: #1)}} 
\newcommand{\enote}[1]{{\color{blue} (Eunou: #1)}}  
\newcommand{\onote}[1]{{\color{magenta} (Ojas: #1)}}

\newcommand{\defeq}{\stackrel{\mathrm{\scriptscriptstyle def}}{=}}
\newcommand{\of}[1]{\left( #1 \right)}
\newcommand{\ofc}[1]{\left\{ #1 \right\}}
\newcommand{\ofb}[1]{{\left[#1\right]}}
\newcommand{\ofk}[1]{{\left\langle#1\right\rangle}}
\newcommand{\abs}[1]{{\left|#1\right|}}
\newcommand{\floor}[1]{{\left\lfloor#1\right\rfloor}}

\newcommand{\lmax}{\lambda_{max}}
\newcommand{\lmin}{\lambda_{min}}

\newcommand{\tr}{\mathrm{Tr}}
\newcommand{\eigs}{\mathrm{eigs}}
\newcommand{\StoqMA}{\class{StoqMA}}

\newtheorem{conjecture}{Conjecture}
\crefname{conjecture}{Conjecture}{Conjectures}
\crefrangelabelformat{conjecture}{#3#1#4--#5#2#6}
\newtheorem{claim}{Claim}
\newtheorem{lemma}{Lemma}
\newtheorem*{lemma*}{Lemma}
\newtheorem{theorem}{Theorem}
\newtheorem{corollary}{Corollary}
\newtheorem{fact}{Fact}

\theoremstyle{definition}
\newtheorem{definition}{Definition}
\newtheorem{remark}{Remark}

\hypersetup{
    colorlinks=true,
    linkcolor=blue,
    filecolor=magenta,      
    urlcolor=cyan,
    citecolor=violet,
}

\begin{document}
\title{Conjectured Bounds for 2-Local Hamiltonians via Token Graphs}
\author{
  Anuj Apte\thanks{University of Chicago}
  \and
  Ojas Parekh\thanks{Sandia National Laboratories}
  \and
  James Sud\footnotemark[1]\;\,\thanks{jsud@uchicago.edu}
}

\date{}
\maketitle

\begin{abstract}
    We explain how the maximum energy of the Quantum MaxCut, XY, and EPR Hamiltonians on a graph $G$ are related to the spectral radii of the \emph{token} graphs of $G$. From numerical study, we conjecture new bounds for these spectral radii based on properties of $G$.  We show how these conjectures tighten the analysis of existing algorithms,
    implying state-of-the-art approximation ratios for all three Hamiltonians. Our conjectures also provide simple combinatorial bounds on the ground state energy of the antiferromagnetic Heisenberg model, which we prove for bipartite graphs.
\end{abstract}

\clearpage
\newpage

\small
\tableofcontents
\normalsize

\clearpage
\newpage

\section{Introduction}\label{sec:intro}
 
Finding ground states of local Hamiltonians is a central problem in quantum computing \cite{kitaev2002}.
Certain well-studied Hamiltonians have been shown to connect to combinatorial graph-theoretic objects \cite{rudolph2002,osborne2006}, enabling the study of 
of these Hamiltonians via techniques from graph theory. In this work, we continue this study by investigating certain $2$-local Hamiltonians via their connection to \emph{token graphs}~\cite{fabila-monroy2012}. As our results concern both quantum computing and spectral graph theory, we provide two introductions. The graph theoretic introduction presents our conjectures on the eigenspectra of token graphs. The quantum computational introduction presents the implications of these conjectures for $2$-Local Hamiltonians and approximation algorithms. Each introduction is largely self-contained, and the reader can choose whichever sections they prefer. A third perspective, which we do not discuss here, arises from the representation theory of the symmetric group (see e.g.~\cite{watts2024,takahashi2023}).

\subsection{A graph theory perspective}
Let $G(V,E,w)$ be a simple graph with vertex set $V=\ofc{1,2,\ldots,n}$, edge set $E$ and positive edge weights $w: E \rightarrow \mathbb{R}_+$. For an integer $1\le k \le n$, the $k$\emph{-th} token graph $F_k(G)$ of $G$ is defined as the simple graph whose vertex set is the $\binom{n}{k}$ $k$-tuples of $V$, where vertices $A$, $B$ are adjacent 
if and only if their symmetric difference is a pair $(a,b)$ such that $a \in A$, $b \in B$, and $(a,b) \in E$, with corresponding edge weight $w_{AB}=w_e$. Token graphs were introduced under the names \emph{k-tuple vertex graphs} \cite{johns1988}, \emph{symmetric k-th power graphs} \cite{audenaert2007}, and \emph{k-tuple vertex graphs} \cite{alavi2002}. Token graphs are also special cases of \emph{Kikuchi} graphs (\cite{manohar2024, trevisan2024}), defined with respect to uniform hypergraphs. Kikuchi graphs were originally introduced by \cite{wein2019} under the name \emph{symmetric difference matrices}. Token graphs were introduced under their current name in \cite{fabila-monroy2012}. The phrase ``token graph'' corresponds to the intuition of placing $k$ indistinguishable tokens on distinct vertices of the graph --- yielding $\binom{n}{k}$ token configurations --- with configurations being adjacent if they can be obtained from each other by moving a single token along an edge in the graph. Note that $F_k(G) \cong F_{n-k}(G)$, so it suffices to consider $1\le k \le \floor{n/2}$.

Let $A(G)$ and $D(G)$ denote the adjacency and degree matrix of a weighted graph $G$. Then, let $L(G)\defeq D(G)-A(G)$ and $Q(G)\defeq D(G)+A(G)$ denote the Laplacian and signless Laplacian matrix of $G$. Let $\lmax\of{\cdot}$, $\lmin\of{\cdot}$, and $\eigs\of{\cdot}$ denote the maximum, minimum, and set of unique eigenvalues of a matrix, respectively. We are interested in studying the eigenspectra of $A(F_k(G))$, $L(F_k(G))$ and $Q(F_k(G))$. These spectra have been the object of much study \cite{alavi1991, rudolph2002, audenaert2007,  jacob2007, barghi2009, alzaga2010, fabila-monroy2012, dalfo2021, reyes2023, reyes2024, barik2024, lew2024, song2024}. However, these studies are limited to the case of unweighted graphs, and we are not aware of current results for the spectra of the signless Laplacian.

For an unweighted graph $G$, let $m$ denote the number of edges. For all unweighted graphs $G$ and for all $1\le k \le \floor{n/2}$, we conjecture

\begin{restatable}{conjecture}{ConjLMax}\label{conj:L_max}
        $\lmax\of{L\of{F_k\of{G}}} \le m+k$.
\end{restatable}

\begin{restatable}{conjecture}{ConjQMax}\label{conj:Q_max}
    $\lmax\of{Q\of{F_k\of{G}}} \le m+k$.
\end{restatable}

\begin{restatable}{conjecture}{ConjAMax}\label{conj:A_max}
    $\lmax\of{A\of{F_k\of{G}}} \le \frac{1}{2}\of{m+k}$.
\end{restatable}

\begin{restatable}{conjecture}{ConjAMin}\label{conj:A_min}
    $\lmin\of{A\of{F_k\of{G}}} \ge -\frac{1}{2}\of{m+k}$.
\end{restatable}

For an arbitrary graph $G$, let $W(G)$ be the sum over edge weights $\sum_{e \in E} w_e$. Let $M_k(G)$ be the weight of the maximum weight matching of $G$ that consists of at most $k$ edges. We show that the above conjectures imply bounds for all \emph{weighted} graphs:

\begin{restatable}{lemma}{LemTokenReductions}\label{lem:token_reductions}
    For all weighted graphs $G$ and all $1 \le k \le \floor{n/2}$
    \begin{align}
        \cref{conj:L_max} &\implies \lmax\of{L\of{F_k\of{G}}} \le W(G)+M_k(G)\,, \label{eq:l_max_reduction}\\
         \cref{conj:Q_max}  &\implies \lmax\of{Q\of{F_k\of{G}}} \le  W(G)+M_k(G)\,, \label{eq:q_max_reduction}\\
         \cref{conj:A_max}  &\implies \lmax\of{A\of{F_k\of{G}}} \le  \frac{1}{2}\of{W(G)+M_k(G)}\,, \label{eq:a_max_reduction}\\
         \cref{conj:A_min}  &\implies \lmin\of{A\of{F_k\of{G}}} \ge -\frac{1}{2}\of{W(G)+M_k(G)}\,\label{eq:a_min_reduction}.
    \end{align}
\end{restatable}

We also analyze the relationship between the spectra of token graphs for different $k$. It is shown in \cite{dalfo2021} that the spectrum of $L(F_k(G))$ is contained in the spectrum of $L(F_{k+1}(G))$, i.e. $\eigs\of{L(F_k(G))} \subseteq \eigs\of{L(F_{k+1}(G))}$. However, this is not true for the signless Laplacian and adjacency matrix. This fact was known for the adjacency matrix, and we provide a counterexample for the signless Laplacian. In contrast, we do observe containment of \emph{maximum} eigenvalues, and for all $G$ and $1\le k < \floor{n/2}$ we conjecture  

\begin{restatable}{conjecture}{ConjQMonotonic}\label{conj:Q_monotonic}
    $\lmax\of{Q\of{F_{k}\of{G}}} \le \lmax\of{Q\of{F_{k+1}\of{G}}}$.
\end{restatable}

\begin{restatable}{conjecture}{ConjAMonotonic}\label{conj:A_monotonic}
    $\lmax\of{A\of{F_{k}\of{G}}} \le \lmax\of{A\of{F_{k+1}\of{G}}} $.
\end{restatable}

\cref{conj:A_monotonic} also appears in \cite[Conjecture 3.3]{reyes2024}. We disprove a corresponding conjecture for the minimum eigenvalue of the adjacency matrix by finding a counterexample where $\lmin\of{A\of{F_{k-1}\of{G}}} \ngeq \lmin\of{A\of{F_k\of{G}}}$.

All the presented conjectures were verified by checking all non-isomorphic unweighted graphs and a suite of weighted graphs, both up to order ten. Our code is available at \cite{apte2025a}.

\subsection{A quantum perspective}
A central task in quantum computation is to compute the maximum energy of local Hamiltonians. This is not an easy task: even for $2$-local Hamiltonians, deciding if the maximum energy is above some threshold is $\QMA$-hard \cite{kempe2005, piddock2015}. In this work, we study three specific families of $2$-local Hamiltonians defined on arbitrary weighted graphs. The first two were introduced as Quantum MaxCut (QMC) and the XY Hamiltonian (XY) in \cite{gharibian2019}. The minimization versions of these problems have been studied as the antiferromagnetic quantum Heisenberg $X\!X\!X_{1/2}$ and $XY_{1/2}$ model in statistical mechanics. The third problem is the EPR Hamiltonian (EPR), recently introduced in \cite{king2023}. The minimization version of EPR is the ferromagnetic quantum Heisenberg $X\!X\!Z_{1/2}$ model in statistical mechanics. While QMC and XY are $\QMA$-complete, EPR is in $\StoqMA$ \cite{piddock2015}.

For these Hamiltonians, a key measure of algorithm performance is the \emph{approximation ratio} $\alpha$, defined as the minimum ratio between the energy achieved by the algorithm and the maximum possible energy over all graphs
\begin{equation}
    \alpha = \min_G \,\frac{ALG(G)}{\lambda_{max}(H(G))}\;.
\end{equation}

Importantly, this ratio depends both on the algorithm's performance and our knowledge of the maximum energy. There are two approaches to improving approximation ratios: developing better algorithms that achieve higher energies, and establishing tighter upper bounds on the maximum energy. In this work, we focus primarily on the latter approach, showing how improved upper bounds on maximum energies directly yield better approximation ratios even without changing the underlying algorithms.

There has been a recent surge in research aimed at finding efficient approximation algorithms for QMC, XY, and EPR \cite{gharibian2019, anshu2020, anshu2021, parekh2022,lee2022,king2023,lee2024,huber2024, takahashi2023, kannan2024, marwaha2024, ju2025, gribling2025, apte2025}. Many of these works obtain upper bounds via semi-definite programming (SDP) relaxations of the maximization problem \cite{mazziotti2004, gharibian2019, parekh2021, takahashi2023, watts2024,huber2024}. These relaxations generally are given by hierarchies of SDPs. One popular choice is quantum moment-SoS hierarchy, typically based on Pauli operators. This hierarchy is an instance of the NPA hierarchy~\cite{navascues2008} and is also known as the quantum Lasserre hierarchy~\cite{parekh2021}. Tighter upper bounds may be obtained by higher levels of the hierarchy. However, higher levels correspond to larger SDPs, which become increasingly more complicated to solve and to analyze.

Our work augments SDP upper bounds by exploiting graph-theoretic results. It is known \cite{osborne2006,audenaert2007,ouyang2019} that the QMC and XY Hamiltonians can be expressed as direct sums of the Laplacian and adjacency matrices, respectively, of the token graphs $F_k(G)$ for $k\in\ofc{0,1,\ldots,n}$. We show that the EPR Hamiltonian can similarly be expressed as a direct sum of the \emph{signless Laplacian} matrices of $F_k(G)$. This representation immediately reveals that maximum energy of QMC is at most that of EPR. More importantly, the equivalence allows us to translate our above conjectures on the spectral radii of token graphs into tighter combinatorial upper bounds for the maximum energies of QMC, XY, and EPR. 

\begin{restatable}{lemma}{hBounds}\label{lem:h_bounds}
    For all graphs $G$, the following implications hold
    \begin{align}
        \cref{conj:L_max}&\implies\lmax\of{H^{QMC}(G)} \le W(G) + M(G) \, . \label{eq:qmc_bound}\\
        \cref{conj:Q_max}&\implies\lmax\of{H^{EPR}(G)} \le W(G) + M(G) \, . \label{eq:epr_bound}\\
        \cref{conj:A_min}&\implies\lmax\of{H^{XY}(G)} \le W(G)+ \frac{M(G)}{2} \, .\label{eq:xy_max_bound}\\
        \cref{conj:A_max}&\implies\lmin\of{H^{XY}(G)} \ge -\frac{M(G)}{2}\,. \label{eq:xy_min_bound}
    \end{align}
\end{restatable}

Throughout this work, $M(G)$, the maximum weight matching, plays an integral role. The relation between the maximum energies of 2-local Hamiltonians and matchings has been observed \cite{lee2024}, which observed that the per-edge energy of QMC on any state  obeys \emph{some} constraints of the linear programming relaxations for maximum matchings. Our improved bounds \cref{lem:h_bounds} rely on the conjecture that \emph{all} of the matching LP constraints must be met.

Furthermore, we show how to augment \emph{any} convex relaxation, including recently used quantum moment-SoS SDP hierarchies, to obey the bounds \cref{lem:h_bounds} as well. This may be surprising, since such arguments typically require new or stronger proofs showing that a Hamiltonian eigenvalue bound translates to a relaxation, and we only have conjectured eigenvalue bounds. We introduce \emph{matching-augmented} convex relaxations that iteratively add constraints derived from the matching polytope to the underlying convex relaxation. We only require that such constraints are valid, which follows from our conjectures. While there are an exponential number of such constraints (roughly one for each odd-order subgraph of the original graph), we appeal to a separation oracle to efficiently identify such violated constraints. Consequently, the ellipsoid method gives a polynomial-time algorithm to find augmented SDP solutions obeying the $W(H) + M(H)$ bound for any subgraph $H$ of $G$, which can be directly used for approximation algorithms. 

These upper bounds in turn imply improved approximation ratios from existing algorithms. We summarize these approximation ratios in \cref{tab:intro_approx}.
\begin{table}[H]
    \centering
    \begin{tabular}{c||c|c|c}
          & Existence $\alpha'$ & Efficient $\alpha$ & State-of-art\\
        \hline
        \hline
         QMC &  $5/8 = 0.625$ & $0.614$ & $0.611$ \cite{apte2025}  \\
        \hline
         XY &  $5/7 \approx 0.714$ & $0.674$ & $0.649$ \cite{gharibian2019} \\
         \hline
         EPR & $\frac{1+\sqrt{5}}{4} \approx 0.809$ & $0.809$ & $0.809$ \cite{ju2025}
    \end{tabular}
    \caption{Approximation ratios implied by our conjectures. Existence $\alpha'$ denotes that there exists a tensor product of $1$ and $2$-qubit states achieving approximation ratio $\alpha'$. Efficient $\alpha$ denotes that there exists an efficient algorithm achieving approximation ratio $\alpha$. State-of-art is the current best approximation ratio at the time of the publication of this work.}
    \label{tab:intro_approx}
\end{table}

All of these algorithms prepare simple tensor products of $1$ and $2$-qubit states. The efficient algorithm for QMC is based on the algorithm of \cite{apte2025}, where we augment the second order SDP algorithm with a matching-based separation oracle. The algorithm for EPR is similar to that of \cite{ju2025, apte2025}. The rest of the algorithms are simple modifications of the algorithm in \cite[Lemma 2]{anshu2020}. The approximation ratios for EPR match the state of the art approximation ratios from \cite{ju2025, apte2025}, which require more globally entangled states. 

\paragraph{Energy bounds for physical local Hamiltonians} From a condensed matter perspective, the bounds on the maximum eigenvalues of $2$-local Hamiltonians can be interpreted as lower bounds on the ground state energy of the corresponding physical systems by a sign flip and ignoring identify terms. For instance, maximizing the QMC Hamiltonian on a graph $G$ corresponds to minimizing the antiferrogmagnetic Heisenberg model (AFHM) on the same graph,
\begin{align*}
H^{AFHM} \defeq \sum_{(i,j) \in E} \frac{w_{ij}}{2}(X_iX_j+Y_iY_j+Z_iZ_j).
\end{align*}
\Cref{conj:L_max} implies the following bound through \Cref{eq:qmc_bound}:
\begin{align*}
\lmin\of{H^{AFHM}} \geq -\frac{W(G)}{2} - M(G).
\end{align*}
This bound generalizes known tight bounds (e.g., on unweighted stars) and to the best of our knowledge has not been previously proposed.

\paragraph{Entanglement bounds}
Our conjectured bounds also suggest a deeper connection between matchings and entanglement. The pairwise concurrence (e.g., Equation (1) in \cite{zhang2006}) of a state $\ket{\psi}$ on qubits $i$ and $j$ is 
\begin{align*}
c_{ij} \defeq \max\{0,-\bra{\psi}\text{SWAP}_{ij}\ket{\psi}\} = \max\{0,\bra{\psi}h^{QMC}_{ij}\ket{\psi}-1\},
\end{align*}
where 
\begin{align*}
h^{QMC}_{ij} \defeq \frac{1}{2}(I-X_iX_j-Y_iY_j-Z_iZ_j)
\end{align*}
is a term of the QMC Hamiltonian, $H^{QMC} \defeq \sum_{(i,j)\in E} h^{QMC}_{ij}$. A product state can earn energy at most $1$ on a single term $h^{QMC}_{ij}$. Consequently if $c_{ij} > 0$, there must be some entanglement on $(i,j)$, and $c_{ij} \in [0,1]$ gives a measure of this. Under \Cref{conj:L_max}, \Cref{eq:qmc_bound} implies a new connection between pairwise concurrence and matchings:
\begin{align}\label{eq:concurrence_bound}
    \sum_{(i,j) \in E} w_{ij} c_{ij} \leq M(G),
\end{align}
for all states $\ket{\psi}$ and weighted graphs $G$. The above follows from \Cref{eq:qmc_bound} by considering the subgraph $H$ of $G$ with edge set $E(H) \defeq \{(i,j) \in E \mid \bra{\psi}h^{QMC}_{ij}\ket{\psi} \geq 1\}$:
\begin{align*}
    \sum_{(i,j) \in E(G)} w_{ij} c_{ij} = \sum_{(i,j) \in E(H)} w_{ij} c_{ij} = \bra{\psi}H^{QMC}(H)\ket{\psi} - W(H) \leq M(H) \leq M(G),
\end{align*}
where the equalities follow from the definition of concurrence, the first inequality follows from \Cref{eq:qmc_bound}, and the second inequality holds because $H$ is a subgraph of $G$. In fact \Cref{eq:concurrence_bound} is equivalent to $(c_{ij})_{(i,j)\in E}$ being a convex combination of matchings. 
This connection helps explain why matchings have played a critical role in recent approximation algorithms for QMC, EPR, and XY.

\subsection{Our contributions}
We now collect and summarize our contributions.
\begin{itemize}
    \item We provide a set of new conjectures on the extremal values of the eigenspectra of unweighted token graphs and verify these conjectures for all non-isomorphic graphs up to order ten (\cref{conj:L_max,conj:Q_max,conj:A_max,conj:A_min}).
    \item We show that these conjectures imply novel combinatorial bounds on the extremal values of the eigenspectra of \emph{weighted} token graphs in terms of maximum weight matchings (\cref{lem:token_reductions}). Our proofs of this connection rely critically on polyhedral descriptions of matchings and related problems (\cref{apx:token_reductions}, \cref{apx:ham_reductions}).
    \item We show that by leveraging the equivalence between commonly-studied $2$-local Hamiltonians and eigenspectra of token graphs, the conjectures also imply combinatorial upper bounds on the maximum energies of QMC, XY, and EPR (\cref{lem:h_reductions}).
    \item We show how to efficiently augment any polynomial-time solvable convex relaxation for the $2$-local Hamiltonian problems with additional matching-based constraints (\cref{sec:imp/matching_sos}). Assuming the conjectures hold, this generically enables efficient convex relaxations to also obey upper bounds identical to our spectral bounds, which is important for designing approximation algorithms. 
    \item We show that these upper bounds would tighten the analysis of existing algorithms for QMC, XY, and EPR, leading to either improved approximation ratios and improved certificates from tensor products of $1$ and $2$-qubit states (\cref{sec:imp}).   
    \item The upper bounds for QMC, EPR, and XY have a  condensed matter interpretation as lower bounds on the ground state energy of the corresponding models in statistical mechanics. A stronger version of our conjectures would imply deeper connections between matchings and entanglement in $2$-local Hamiltonians.
\end{itemize}

We visualize our conjectures and their implications in \cref{fig:conjecture_dag}.
\begin{figure}[H]
    \centering
    \includegraphics[width=1\linewidth]{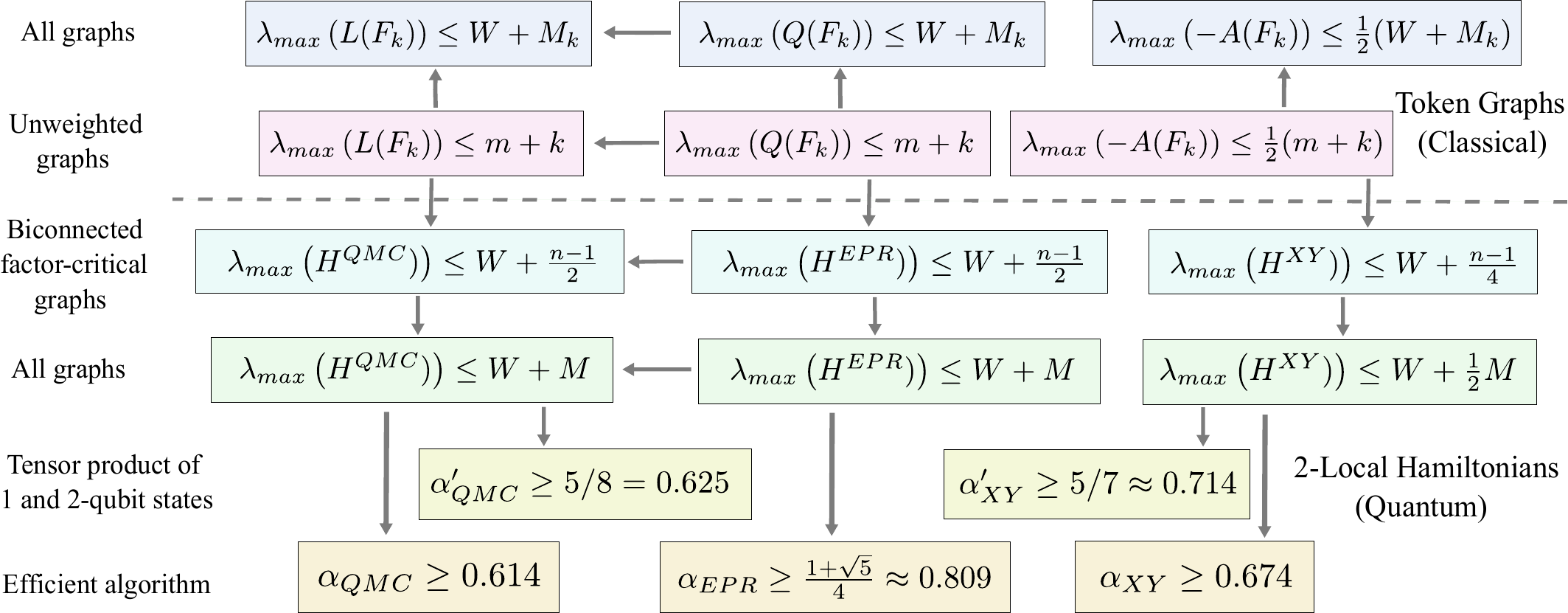}
    \caption{Conjectures in this work and their implications. Equations above the dashed line strictly concern token graphs, while equations below the dashed line strictly concern quantum Hamiltonians. We omit the argument $\of{G}$ in all cells for ease of notation. Following our convention, $\alpha'$ denotes the approximation ratio achieved by a tensor product of $1$ and $2$-qubit states, regardless of whether or not this tensor product can be found efficiently, and $\alpha$ denotes the approximation ratio achieved by an efficient algorithm.}\label{fig:conjecture_dag}
\end{figure}

\paragraph{Note added:}
v2: \cref{conj:L_max,conj:Q_max,conj:A_max,conj:A_min} were proven by \cite{bakshi2026a} after release of this preprint. Thus, all statements in \cref{fig:conjecture_dag} are true for all graphs. Furthermore, we have corrected a small error in the approximation algorithm for QMC from the first version.

v3: \cref{conj:A_cut_bound} was refuted in a personal correspondence from Zhang Wenlin and Haobo Ma. A lean proof is given in \cite{omega2026}.

\section{Preliminaries}\label{sec:prelim}

\subsection{Graph theory}\label{sec:prelim/graph}
Let $G(V,E,w)$ denote a simple graph with vertex set $V$, edge set $E$, and positive edge weights $w: E \rightarrow \mathbb{R}_+$. In this work, we always take $V=\ofb{n} \defeq \{1,2,\dots,n\}$ and $m\defeq \abs{E}$. We let $W(G) \defeq \sum_{(i,j) \in E} w_{ij}$. When $G$ is inferred by context, we simply write $W$. Define $N(v)$ as the set of neighbors of a vertex $v \in V$. For convenience, we index edges by either $e$ or $(i,j)$, or simply $ij$ in subscripts.

Given a subset $S\subseteq V$, let $G(S)$ ($E(S)$) denote the graph (edges) \emph{induced} by the vertices in $S$. Let $\delta(S)$ be the set of edges in $G$ that cross from $S$ to $V \setminus S$, and $Cut(S)$ be the sum of the weights of the edges $\delta(S)$. We denote the maximum of $Cut(S)$ over all $S$ by $C(G)$, i.e., $C(G)$ is the solution to the Maximum Cut (MaxCut) problem on $G$. We sometimes also use $C(G)$ to denote the edges that are cut in the maximum cut. When the graph is clear from context we drop the argument $G$. Computing the maximum cut of a graph in general is $\NP$-hard \cite{khot2007}. 

An \emph{matching} of a graph $G(V,E,w)$ is a subset of edges $R \subseteq E$ such that the number of edges in $R$ incident to any vertex $i$ is at most one. The \emph{incidence vector} of a matching is a vector in $\ofc{0,1}^E$ that is $1$ if $e \in M$ and $0$ otherwise.
The \emph{weight} of a matching $R$ is defined as as the sum of the weights of edges in the matching. We define $M(G)$ to be the maximum total weight of any matching of $G$. For convenience, we sometimes let $M(G)$ also denote the set of edges in the maximum weight matching. Let $M_k(G)$ denote the maximum weight matching under the restriction that $|M|\leq k$. When the graph is clear from context we sometimes drop the argument $G$. Both $M(G)$ and $M_k(G)$ can be computed in polynomial time~\cite{edmonds1965, araoz1983}.

Let $D(G)$ and $A(G)$ denote the degree matrix and the adjacency matrix of a graph $G$. Then, let $L(G)=D(G)-A(G)$ and $Q(G)=D(G)+A(G)$ denote the Laplacian and signless Laplacian matrices of $G$. Let $\lmax(\cdot), \lmin(\cdot)$, and $\eigs(\cdot)$ describe the maximum eigenvalue, minimum eigenvalue, and the set of unique eigenvalues of a matrix. Occasionally, we refer to $\lmax\of{L\of{G}}$, $\lmax\of{Q\of{G}}$, and $\lmax\of{A\of{G}}$ as the \emph{spectral radii} of $G$.

Let $C_n$, $P_n$, and $K_n$ denote the unweighted cycle, path, and complete graphs on $n$ vertices, respectively. Let $S_m$ denote the unweighted star graph with $m$ edges (and $m+1$ nodes). We let $K_{a,b}$ denote the complete bipartite graph with partition sizes $a$, $b$. We now define token graphs using the notation introduced in \cite{fabila-monroy2012}

\begin{definition}[Token graphs]\label{def:token_graphs}
    Given a graph $G(V,E,w)$ and some integer $1\le k < n$, let the $k$\emph{-th token graph} $F_k(G)$ be a weighted simple graph defined as follows:
    \begin{itemize}[label=\raisebox{0.4ex}{\scalebox{0.75}{$\bullet$}}]
        \item Vertices: vertices are $\binom{\ofb{n}}{k}$, the set of $k$-tuples of the set $\ofb{n}$, which contains $\binom{n}{k}$ elements.
        \item Edges: vertices $A$ and $B$ are connected by an edge if and only if their symmetric difference $A\triangle B =\ofc{a,b}$, where $a\in A$, $b\in B$, and $(a,b)\in E$.
        \item Weights: the weight of edge $(A,B)$ is $w_{ab}$, where $a$ and $b$ are defined as above.
    \end{itemize}
\end{definition}

Note that $F_k(G) \cong F_{n-k}(G)$. As such the unique token graphs of some graph $G$ arise from $1\le k \le \floor{n/2}$. We show an example of the three unique token graphs for the graph $P_6$ in \cref{fig:path_token_graphs}.  Previous literature on token graphs has been limited to the unweighted case \cite{dalfo2021,reyes2023,ouyang2019,lew2024, dalfo2025}. In this work, we also introduce and study token graphs of \emph{weighted} graphs.

\begin{figure}[H]
    \centering
    \includegraphics[width=0.75\linewidth]{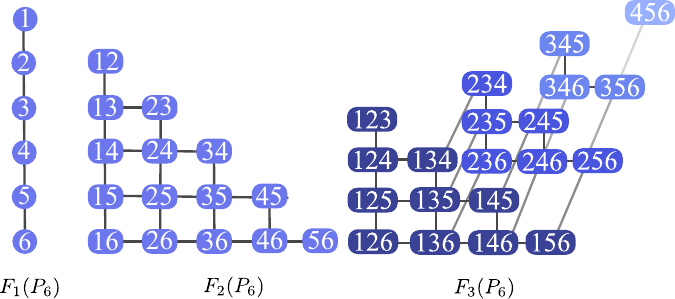}
    \caption{The three unique token graphs of the six node path $P_6$. The white numbers inside the token graph vertices denote subsets of the vertices of $P_6$, concatenated into a single string of integers for conciseness. Note that $F_1(P_6) \cong F_5(P_6)$, $F_2(P_6) \cong F_4(P_6)$, and $F_3(P_6) \cong F_4(P_6)$. The shading on $F_3(P_6)$ is to emphasize that the graph is a subgraph of a three-dimensional grid, with different vertical layers corresponding to different shades.} 
    \label{fig:path_token_graphs}
\end{figure}

\subsection{Quantum computation}\label{sec:prelim/quantum}
The \emph{Bell basis} consists of the four maximally entangled two-qubit states, given by
\begin{align*}
    \ket{\phi^+} &= \frac{1}{\sqrt{2}} \left( \ket{00} + \ket{11} \right), \quad\ket{\phi^-} = \frac{1}{\sqrt{2}} \left( \ket{00} - \ket{11} \right), \\
    \ket{\psi^+} &= \frac{1}{\sqrt{2}} \left( \ket{01} + \ket{10} \right), \quad\ket{\psi^-} = \frac{1}{\sqrt{2}} \left( \ket{01} - \ket{10} \right).
\end{align*}
We refer to $\ket{\phi^+}$ as the \emph{EPR pair} and to $\ket{\psi^-}$ as the \emph{singlet state}. We refer to the single qubit maximally mixed state as $\rho_{mix}$.

Given a graph $G(V, E, w)$ and a $2$-local Hamiltonian term $h$, define the $n$-qubit Hamiltonian 
\begin{align}\label{eq:2_local_ham_defn}
    H(G) \defeq \sum_{(i,j) \in E(G)} w_{ij} \cdot h_{ij}\,,
\end{align}
where $h_{ij}$ is the tensor product of the local term $h$ applied on qubits $(i,j)$ with the identity on all remaining qubits. We denote by $\lmax\of{H(G)}$, $\lmin\of{H(G)}$, $\eigs\of{H(G)}$ the maximum, minimum, and set of unique eigenvalues of a Hamiltonian $H(G)$.

We consider three choices of local terms, which we refer to as Quantum MaxCut (QMC), the EPR Hamiltonian (EPR), and the XY Hamiltonian (XY)
\begin{align}
    h^{QMC}_{ij} &\defeq \frac{1}{2} \left( I_i I_j - X_iX_j - Y_iY_j - Z_iZ_j \right) = 2 \,\ket{\psi_-}_{ij}\bra{\psi_-}_{ij}\,, \label{eq:qmc_defn}\\
    h^{XY}_{ij}  &\defeq \frac{1}{2} \left(I_iI_j  -X_iX_j - Y_iY_j \right) = \frac{1}{2} \, I_i I_j +  \,\ket{\psi_-}_{ij}\bra{\psi_-}_{ij}-\,\ket{\psi_+}_{ij}\bra{\psi_+}_{ij}\,,\label{eq:xy_defn} \\
    h^{EPR}_{ij} &\defeq \frac{1}{2} \left( I_i I_j + X_iX_j - Y_iY_j + Z_iZ_j\right) = 2 
    \ket{\phi_+}_{ij}\bra{\phi_+}_{ij}\,.\label{eq:epr_defn}
\end{align}
It is clear that terms of QMC and EPR are positive semidefinite (PSD) rank-$1$ projectors onto Bell basis states. In particular, QMC projects onto the \emph{singlet} state and EPR projects onto the EPR pair. In contrast, terms of XY are rank-$2$ projectors onto the Bell basis, offset by a constant, and are \emph{not} PSD.

\subsection{Approximation algorithms}\label{sec:prelim/apx}
We judge an algorithm by its \emph{approximation ratio}. We follow the definition in \cite{gharibian2019} for $2$-local Hamiltonians. Suppose we can find an efficiently computable upper bound $\lmax(H(G)) \leq u(G)$ for any graph $G$. Then, if we have some algorithm $ALG$ which obtains energy $ALG(G)$ on a graph $G$, the approximation ratio $\alpha$ is at least
\begin{align*}
    \alpha \ge \min_G \frac{ALG(G)}{\lmax(H(G))} \ge  \min_G \frac{ALG(G)}{u(G)}\,.
\end{align*}
In particular, the approximation is well-defined if $ALG(G)\geq0$, $\lmax(H(G))>0$ for all $G$. This can be easily confirmed; the maximally mixed state obtains energy $ W(G)/2 > 0$ for all Hamiltonians considered in \cref{sec:prelim/quantum}. 

\section{Hamiltonian and token graph equivalence}\label{sec:equivalence}
We now formalize the connection between the $2$-Local Hamiltonians introduced in \cref{sec:prelim/quantum} and token graphs as described \cref{sec:prelim/graph}. The main technique in showing the equivalence is to write each Hamiltonian $H(G)$ in a block-diagonal form, and to demonstrate that each block is equivalent to a Laplacian, adjacency, or signless Laplacian matrix on a token graph of $G$.

\subsection{QMC and the Laplacian matrix}\label{sec:equivalence/qmc}
We start by considering QMC, where the connection has already been observed in \cite{osborne2006, ouyang2019}. We first expand \cref{eq:qmc_defn} in the computational basis

\begin{align}\label{eq:qmc_computational_basis}
    H^{QMC}(G) = \sum_{(i,j) \in E} \of{\ket{01}\bra{01}}_{ij} +  \of{\ket{10}\bra{10}}_{ij} -  \of{\ket{01}\bra{10}}_{ij} -  \of{\ket{10}\bra{01}}_{ij}.
\end{align}

In this form, we see that $H^{QMC}(G)$ preserves Hamming weight in the computational basis. Thus, the Hamiltonian can be block-diagonalized into $n+1$ blocks, each corresponding to a Hamming weight-$k$ sector for $k=\ofc{0,1,2,\ldots,n}$. We may then analyze the action of $H^{QMC}(G)$ in a fixed Hamming weight sector. For a subset of vertices $X \subseteq V$ of size $k$, let $x\in \ofc{0,1}^n$ be the length-$n$ bitstring where $x_i = 1$ if and only if $i \in X$. Let $\ket{x}$ be the quantum state encoding bitstring $x$ in the computational basis. Let $h(x)$ denote Hamming weight of $x$. Then, the $\binom{n}{k}$ states $\ofc{\ket{x}\!:  x\in \ofc{0,1}^n, h(x)=k}$ span the Hilbert space of $n$-qubit states with fixed Hamming weight $k$. With this notation, the action \cref{eq:qmc_computational_basis} on a fixed basis state $\ket{x}$ can be expressed as
\begin{align*}
    H^{QMC}(G) \ket{x} = \sum_{(i,j)\in E} w_{ij}\,\mathbf{1}\{x_i \neq x_j\}\big(\ket{x}\bra{x} - \ket{x^{(i,j)}}\bra{x}\big)\,, 
\end{align*}
where $x_i$, $x_j$ denote the bits of $x$ at index $i$ or $j$, and $x^{(i,j)}$ denotes the bitstring $x$ with $x_i$ and $x_j$ interchanged. We may then compute the inner products
\begin{align}\label{eq:qmc_overlaps}
    M^{QMC}_{x,y}(G)\defeq\braket{y|H^{QMC}(G)|x} = \begin{cases}
        \sum_{(i,j)\in E} \, w_{ij} \,\mathbf{1}\ofc{x_i \ne x_j}, & y=x,\\
        -w_{ij}  & y = x^{(i,j)}.
    \end{cases}
\end{align}

In this form, we can see a diagonal entry $M^{QMC}_{x,x}(G)$ is the sum of the weights of edges $(i,j)$ for which $x_i \neq x_j$. The off-diagonal elements are nonzero if and only if the symmetric difference between the sets $X$ and $Y$ corresponding to $x$ and $y$ is some pair $(i,j)\in E$. If this is the case, the edge in the token graph has weight $-w_{ij}$. By referring to \cref{def:token_graphs}, we can see that this inner product matrix $M^{QMC}_{x,y}(G)$ corresponds exactly to the Laplacian matrix of the \emph{$k$-th} token graph of $G$. This result holds for any $k$, so the matrix $H^{QMC}$ in the computational basis consists of $(n+1)$ blocks indexed by Hamming weight $k$, where each block corresponds to the Laplacian $L(F_k(G))$. 

\Needspace{10\baselineskip}
\begin{factbox}
\begin{fact}\label{fact:qmc_L_equivalence}
    For any graph $G$, the QMC Hamiltonian $H^{QMC}(G)$ is equivalent to a direct sum over the Laplacian matrices $L(F_k(G))$ of the token graphs $(F_k(G))$ for $0 \leq k \leq n$. Furthermore, we have
    \begin{align*}
        \eigs(H^{QMC}(G)) = \bigcup_{0 \leq k \leq \floor{\frac{n}{2}}} \eigs(L(F_k(G)).
    \end{align*}
\end{fact}
\end{factbox}

\subsection{XY and the adjacency matrix}\label{sec:equivalence/xy}
We now consider the XY Hamiltonian. Expressing the Hamiltonian \cref{eq:qmc_defn} in the computational basis, we obtain 

\begin{align*}
    H^{XY}(G) = \sum_{(i,j) \in E} \frac{1}{2}I_{ij} - \,\of{\ket{01}\bra{10}}_{ij} -  \of{\ket{10}\bra{01}}_{ij}.
\end{align*}

This is exactly equivalent to the second two terms of \cref{eq:qmc_computational_basis}, shifted by the identity. Thus, we can borrow the analysis for QMC, dropping suitable terms. Following the notation for QMC, we arrive at 

\begin{align*}
    H^{XY}(G) \ket{x} = \frac{W(G)}{2}I-\sum_{(i,j)\in E} \mathbf{1}\{x_i \neq x_j\}\, \ket{x^{(i,j)}}\bra{x}\,. 
\end{align*}

We may again define the inner products

\begin{align*}
    M^{XY}_{x,y}\defeq\braket{y|\of{H^{XY}-\frac{W(G)}{2}I}|x} = \begin{cases}
        0, & y=x,\\
        -w_{ij}  & y = x^{(i,j)}.
    \end{cases}
\end{align*}

This overlap matrix is exactly equivalent to that of \cref{eq:qmc_overlaps}, albeit without the diagonal terms. By definition, the off-diagonal of a Laplacian is the negative of the adjacency matrix.

\Needspace{10\baselineskip}
\begin{factbox}
\begin{fact}\label{fact:xy_A_equivalence}
    The XY Hamiltonian $H^{XY}(G)$ is equivalent to a direct sum over scaled and shifted adjacency matrices $\frac{W}{2}I-A(F_k(G))$ of the token graphs $(F_k(G))$ for $0 \leq k \leq n$. Furthermore, we have
    \begin{align*}
        \eigs(H^{XY}(G)) = \bigcup_{0 \leq k \leq \floor{\frac{n}{2}}} \eigs\of{\frac{W}{2} I - A\of{F_k\of{G}}}.
    \end{align*}
\end{fact}
\end{factbox}

\subsection{EPR and the signless Laplacian matrix}\label{sec:equivalence/epr}

For EPR, we observe that the unitary $U \defeq \otimes_{j\in V} U_j$, $U_j \defeq \sqrt{X_j}$ applies the following mapping
\begin{align}\label{eq:U_transform}
    U X_iX_j U^\dagger = X_iX_j, \quad
    U_j Y_iY_j U^\dagger = Z_iZ_j, \quad 
    U_j Z_jZ_j U^\dagger = Y_iY_j.
\end{align}
Thus, we have from \cref{eq:epr_defn} and linearity over edges that $U$ accomplishes the transformation
\begin{align*}
    U \of{H^{EPR}(G)} U^\dagger &= \frac{1}{2}\sum_{(i,j)\in E}\,w_{ij}\of{ I_iI_j + X_i X_j + Y_i Y_j - Z_i Z_j}, \\
        &=  \sum_{(i,j)\in E} \of{\ket{01}\bra{01}}_{ij} +  \of{\ket{10}\bra{10}}_{ij} + \of{\ket{01}\bra{10}}_{ij} + \of{\ket{10}\bra{01}}_{ij}.
\end{align*}

Once again, this Hamiltonian is identical to \cref{eq:qmc_computational_basis}, albeit with all terms positive. Borrowing the analysis for QMC, we then arrive at \cref{eq:qmc_overlaps} but with all terms positive. Taking the absolute value of all terms in the Laplacian yields the signless Laplacian. 

\Needspace{10\baselineskip}
\begin{factbox}
\begin{fact}\label{fact:epr_Q_equivalence}
    The EPR Hamiltonian $H^{EPR}(G)$ is equivalent under unitary transformation to a direct sum over the signless Laplacian matrices $Q(F_k(G))$ of the token graphs $(F_k(G))$ for $0 \leq k \leq n$. Furthermore, we have
    \begin{align*}
        \eigs(H^{EPR}(G)) = \bigcup_{0 \leq k \leq \floor{\frac{n}{2}}} \eigs(Q(F_k(G)).
    \end{align*}
\end{fact}
\end{factbox}

\subsection{Applications}\label{sec:equivalence/app}
Using the equivalence between Hamiltonians and token graphs, we may then immediately apply useful results from spectral graph theory. We highlight some results here, and provide a few more in \cref{apx:other_app}.

\subsubsection{Bipartite graphs}
For bipartite graphs, we observe the following

\begin{lemma}[Folklore, see \cite{brouwer2012}]\label{lem:bipartite_token_symmetry}
    For any bipartite graph $G$,
    \begin{align*}
        \eigs(L(G))&=\eigs(Q(G)), \\
        \eigs(A(G))&=\eigs(-A(G)).
    \end{align*}
\end{lemma}

\begin{lemma}[Proposition 12 \cite{fabila-monroy2012}]\label{lem:bipartite_graph_has_bipartite_token}
    If $G$ is bipartite, $F_k(G)$ is bipartite for any $1\le k \le n$.
\end{lemma}

\cref{lem:bipartite_token_symmetry} and \cref{lem:bipartite_graph_has_bipartite_token} can be readily combined with \cref{fact:qmc_L_equivalence}, \cref{fact:epr_Q_equivalence} and \cref{fact:xy_A_equivalence} to obtain

\begin{lemma}\label{lem:bipartite_equivalence}
    For any bipartite graph $G$ 
    \begin{align*}
        \eigs\of{H^{QMC}(G)} &= \eigs\of{H^{EPR}(G)}, \\
        \eigs\of{H^{XY}(G)-\frac{W}{2}I} &= \eigs\of{-H^{XY}(G)+\frac{W}{2}I}.
    \end{align*}
\end{lemma}

The equivalence between EPR and QMC for bipartite graphs was first observed by \cite{king2023}. 

\subsubsection{Relation of QMC and EPR}
We may also relate QMC and EPR beyond bipartite graphs. We first show

\begin{lemma}\label{lem:L_less_than_Q}
    For any graph $G$
    \begin{align*}
        \frac{\lmax\of{Q(G)}}{2} \le \lmax\of{L(G)} \le \lmax\of{Q(G)}.
    \end{align*}
\end{lemma}
\begin{proof}
    The right inequality is given in \cite{merris1998} (see \cite{shu2002}). We now show the left inequality. Given a graph $G$, let $i^* \in \ofb{n}$ be the vertex maximizing $d_i \defeq \sum_{j \in N(i)}w_{ij}$. Let $v$ be a vector that is $1$ exactly at the index corresponding to $i^*$ and $0$ everywhere else. Then $v^\dagger (D \pm A) v = d_{i^*}$. Thus, both $\lmax\of{L(G)}\ge d_{i^*}$ and $\lmax\of{Q(G)}\ge d_{i^*}$. By Gershgorin's circle theorem, we have $Q(G) \le 2 d_{i^*}$, so $\frac{\lmax\of{Q(G)}}{2} \le \lmax\of{L(G)}$. 
\end{proof}

We then can state
\begin{lemma}\label{lem:L_less_than_Q_token}
    For any graph $G$ and any $1\le k \le \floor{n/2}$
    \begin{align*}
        \frac{\lmax\of{F_k\of{Q(G)}}}{2} \le \lmax\of{F_k\of{L(G)}} \le \lmax\of{F_k\of{Q(G)}}.
    \end{align*}
    
    Furthermore, the right inequality is tight for bipartite graphs and the left inequality is asymptotically tight for complete graphs.
\end{lemma}
\begin{proof}
    The inequalities follow immediately from \cref{lem:L_less_than_Q}. Tightness of the right inequality follows from \cref{lem:bipartite_equivalence}. 
    To see tightness on the left inequality, consider complete graphs $K_n$ as $n \rightarrow \infty$. Using \cref{lem:H_eigenvalues_complete} with $i=k$ for QMC and $i=0$ for EPR we have
    \begin{align*}
        \lim_{n \rightarrow \infty}\frac{\lmax\of{Q\of{K_n}}}{\lmax\of{L\of{K_n}}} = \lim_{n \rightarrow \infty}\frac{2k\of{n-k}}{k\of{n-k+1}}=2,
    \end{align*}
    for any $1\le k \le \floor{n/2}$.
\end{proof}

\cref{lem:L_less_than_Q_token} can be readily combined with \cref{fact:qmc_L_equivalence} and \cref{fact:epr_Q_equivalence} to obtain

\begin{corollary}\label{cor:qmc_less_than_epr}
    For any graph $G$
    \begin{align*}
         \frac{\lmax\of{H^{EPR}(G)}}{2} \le \lmax\of{H^{QMC}(G)} \le \lmax\of{H^{EPR}(G)}.
    \end{align*}
    Furthermore, the right inequality is tight for all bipartite graphs and the left inequality is asymptotically tight for complete graphs.
\end{corollary}

\subsubsection{Location of maximum eigenstates}
The relationship between the eigenspectra of the token graphs of $G$ for different $k$ yields useful insights into the maximum-energy eigenstates of the corresponding Hamiltonians. First consider the following lemma

\begin{restatable}{lemma}{LemLaplacianSpectrumContainment}\label{lem:laplacian_spectrum_containment}
    For any graph $G(V,E,w)$ with $n$ nodes and any $1\leq k \leq \floor{n/2}-1$
    \begin{align}
        \eigs\of{L(F_k(G))} \subseteq \eigs\of{L(F_{k+1}(G))}\,.
    \end{align}
\end{restatable}
This lemma was proven in \cite[Theorem 4.1]{dalfo2021} for unweighted token graphs. We prove it for the weighted case in \cref{apx:other_app/spectrum_containment}. Combining \cref{lem:laplacian_spectrum_containment} with \cref{fact:qmc_L_equivalence} yields

\begin{lemma}\label{lem:qmc_ham_weight_n_over_2_optimal}
     For any graph $G$ with $n$ nodes, there always an exists a maximum eigenvalue for $H^{QMC}(G)$ with an associated eigenstate supported only on bitstrings with Hamming weight $\floor{n/2}$.
\end{lemma}

This observation is implicit in \cite[Appendix D]{yordanov2020}. In particular, it means that any algorithms to find high energy states for QMC can focus on states supported only on Hamming weight $\floor{n/2}$ bitstrings. This effectively reduces the size of the Hilbert space from $2^n$ to $\binom{n}{n/2}$, which provides a $1/\sqrt{n}$ reduction in the large size limit. This result does not rule out the possibility of maximum energy eigenstates in smaller Hamming weight subspaces. Indeed, the unweighted star graph always has a maximum energy eigenstate with Hamming weight $1$ \cite{lieb1962, anshu2020}.

\subsubsection{Monogamy of entanglement for XY on a star}
An important contribution in determining upper bounds for the maximum eigenvalue of QMC and EPR is \emph{monogamy of entanglement} on a star graph: 

\begin{lemma}\label{lem:l_star_bound}
    For a star graph $S_m$ with $m\ge1$ edges, $n=m+1$ nodes and any $1 \le k \le \floor{n/2}$
    \begin{align*}
        \lambda_{max}\of{L\of{F_k\of{S_m}}} = \lambda_{max}\of{Q\of{F_k\of{S_m}}} = m+1.
    \end{align*}
\end{lemma}
\begin{proof}
    It is shown in \cite{anshu2020} that $\lambda_{max}\of{H^{QMC}\of{S_m}} = m+1$. By \cref{fact:qmc_L_equivalence}, this means that the bound holds for all $k$, as $\lambda_{max}\of{H^{QMC}\of{S_m}}$ corresponds to the maximum of $\lambda_{max}\of{L\of{F_k\of{S_m}}} \le m+1$ over all $k$. By \cref{lem:bipartite_equivalence}, these statement hold when replacing $L$ with $Q$ and $H^{QMC}(G)$ with $H^{EPR}(G)$.
\end{proof}

We also may derive similar monogamy of entanglement results for XY on a star. The full spectrum of the token graphs of stars is characterized in \cite{dalfo2021} via methods in \cite{dalfo2020} (though the results for adjacency matrix are not stated explicitly). We only require bounds on extremal spectra, which we state below

\begin{restatable}{lemma}{StarAMax}\label{lem:star_A_max}
    For a star graph $S_m$ and any $1 \le k \le \floor{n/2}$
    \begin{align*}
        \lmax\of{A\of{F_k\of{S_m}}} = \sqrt{k(m+1-k)}.
    \end{align*}
\end{restatable}
The proof is given in \cref{apx:token_reductions}. By \cref{lem:bipartite_token_symmetry}, this also yields the bound 
\begin{corollary}\label{cor:star_A_min}
    For a star graph $S_m$ and any $1 \le k \le \floor{n/2}$
    \begin{align*}
        \lmin\of{A\of{F_k\of{S_m}}} \ge -\sqrt{k(m+1-k)}.
    \end{align*}
\end{corollary}
\cref{cor:star_A_min}, and \cref{lem:bipartite_token_symmetry} then yield
\begin{corollary}\label{cor:star_xy_monogamy}
    For a star graph $S_m$
    \begin{align*}
         \lmax\of{H^{XY}\of{S_m}} =  \begin{cases}
            \frac{m}{2}+\sqrt{\frac{m^2+2m}{4}}, & m \text{ even} \\ 
            m + \frac{1}{2}, & m \text{ odd} 
        \end{cases}
    \end{align*}
\end{corollary}
\begin{proof}
    The proof of the second equality comes from maximizing \cref{cor:star_A_min} over $k$. This yields $k=(m+1)/2$ for even $m$ and $k=m/2$ for odd $m$. These values are then plugged into \cref{fact:xy_A_equivalence}, along with $W_{S_m}=m$.
\end{proof}

\section{Conjectures}\label{sec:conj}
We now present our conjectures on the extremal spectra of token graphs and their respective Hamiltonians under the equivalence in \cref{sec:equivalence}.

\subsection{Token graphs}\label{sec:conj/token}
Our conjectures for token graphs do not require any prerequisite knowledge of quantum computation, and we believe them to be of independent interest to the spectral graph community. We restate them for convenience convenience: for all  unweighted graphs $G$ and for all $1\le k\le \floor{n/2}$, we conjecture

\ConjLMax*

\ConjQMax*

\ConjAMax*

\ConjAMin*

These bounds are in many cases tighter than those currently found in the literature. For example, a different upper bound for $\lmax\of{L\of{F_k\of{G}}}$ is presented in \cite[Theorem 1.2]{lew2024}
\begin{align}\label{eq:lew_upper_bound}
    \lmax\of{L\of{F_k\of{G}}} &\le k \, \lmax \of{L(G)}, \quad\quad \forall \,\,0\le k \le \left\lfloor \frac{n}{2} \right\rfloor.
\end{align}
This upper bound is \emph{multiplicative} with respect to the maximum eigenvalue of the Laplacian of $G$, while ours is \emph{additive} with respect to the number of edges in $G$. As an example, consider the star graph $S_m$. Our bound states that $\lmax\of{F_k(S_m))} \le m+k$, while \cref{eq:lew_upper_bound} yields $k(m+1)=mk+k$.

For the maximum Laplacian eigenvalue, we show it suffices to prove the conjecture for \emph{triangle free graphs}
\begin{lemma}\label{lem:token_remove_triangles}
    Suppose there exists some unweighted graph $G$ where $\lmax\of{L\of{F_k\of{G}}} > m+k$. Then there exists a triangle-free graph $G'$ with $m'$ edges such that $\lmax\of{L\of{F_k\of{G'}}} > m'+k$.
\end{lemma}
\begin{proof}
    If $G$ is triangle-free, we are done. Otherwise, decompose $G$ into a graph $H$ with $m_H$ edges and the triangle $\Delta$. Then by the triangle inequality
    \begin{align*}
        \lmax\of{L\of{F_k\of{H}}} &\ge \lmax\of{L\of{F_k\of{G}}} - \lmax\of{L\of{F_k\of{\Delta}}}\\
        &> m + k - \lmax\of{L\of{F_k\of{\Delta}}}.
    \end{align*}

We can check $\lmax\of{L\of{F_k\of{\Delta}}} \le 3$. Thus, $\lmax\of{L\of{F_k\of{H}}} > m_H+k$, where $m_H=m-3$ is the number of edges in $H$. We may repeat this argument now starting with $H$ until we have removed all triangles. 
\end{proof} 

We additionally investigate the relationship of the maximum eigenvalues of the signless Laplacian and adjacency matrices for token graphs of differing $k$. While it is known that the spectrum of the Laplacian of the \emph{$k$-th} token graphs are contained in the  \emph{$(k+1)$-th} token graph (\cref{lem:laplacian_spectrum_containment}), similar results are not true for the adjacency matrix or signless Laplacians. The graph $K_4$ provides a simple counterexample: $\eigs(F_1(A(K_4)))=\ofc{-1,3}$ while $\eigs(F_2(A(K_4)))=\ofc{-2,0,4}$, and $\eigs(F_1(Q(K_4)))=\ofc{2,6}$ while $\eigs(F_2(Q(K_4)))=\ofc{2,4,8}$. Counterexamples for the adjacency matrix were previously observed (see below Equation~2 of \cite{rudolph2002}).

A weaker implication of \cref{lem:laplacian_spectrum_containment} is that $\lmax\of{L\of{F_{k}\of{G}}} \le \lmax\of{L\of{F_{k+1}\of{G}}}$. While the eigenspectra of the adjacency and signless Laplacian matrices do not follow ordered containment, we analogously conjecture that their \emph{maximum} eigenvalues may be ordered

\ConjQMonotonic*

\ConjAMonotonic*

We find, however, that $\lmin\of{A\of{F_{k-1}\of{G}}} \ge \lmin\of{A\of{F_k\of{G}}}$ is not true. A counterexample is given in \cref{apx:numerics/adj_containment}.

Our conjectures are all supported by numerical verification on the following set of graphs
\begin{itemize}
    \item All non-isomorphic unweighted graphs up to order ten (totaling over $12$ million graphs).
    \item A suite of weighted graphs consisting of
    \begin{itemize}
        \item $10,000$ Erdős--Rényi random graphs per order ($n=3,\ldots,10$) with edge probability $p \sim \text{Uniform}(0,1)$ and edge weights $w \sim \text{Uniform}(0,1)$
        \item $5,000$ complete graphs per order with edge weights drawn from three different distributions: uniform$(0,1)$, exponential with $\lambda=1$, and exponential with $\lambda=10$.
    \end{itemize}
\end{itemize}
We implemented the computations in \texttt{Python}, using the \texttt{nauty} and \texttt{Traces} packages \cite{mckay2013} for efficient graph generation,  \texttt{NumPy} and \texttt{SciPy} \cite{virtanen2020, harris2020} for linear algebra operations and eigenvalue computations, and \texttt{NetworkX} \cite{hagberg2008} for graph manipulation and structural analysis. The codebase and generated dataset used to test the conjectures is available at \cite{apte2025a}.

\subsubsection{Extension to weighted graphs}
We show that the set of conjectures in \cref{sec:conj/token} actually yield stronger combinatorial conjectures on arbitrary (possibly weighted) graphs. We restate the lemma for convenience and defer the proof to \cref{apx:token_reductions}

\LemTokenReductions*

\subsection{Local Hamiltonians}\label{sec:conj/hams}
Using the connection between token graphs and $2$-local Hamiltonians established in \cref{sec:equivalence}, we can translate our conjectured bounds on the eigenspectra of token graphs in \cref{sec:conj/token} to bounds on the maximum energy of QMC, XY, and EPR, as introduced in \cref{sec:prelim/quantum}

\hBounds*

These corollaries can be easily confirmed by the equivalence between the QMC, EPR, and XY Hamiltonians with the union of the matrices $L(F_k(G))$, $Q(F_k(G))$, and $-A(F_k(G))$ as described in \cref{sec:equivalence}. The bounds \cref{conj:L_max} and \cref{conj:Q_max} have been previously hypothesized, first in \cite{lee2024} and later in \cite{jorquera2024, gribling2025, apte2025}. However, we formally propose the conjectures here.

\begin{remark}
    The conclusion in \cref{eq:qmc_bound} for QMC is shown to be tight for complete bipartite graphs in \cite{lieb1962,rademaker2019,takahashi2023}. The conclusion in \cref{eq:epr_bound} for EPR is therefore tight for complete bipartite graphs via \cref{lem:bipartite_equivalence}. It can be easily confirmed that the bound is tight for EPR on complete graphs from \cref{apx:other_app/complete}.
\end{remark}

\subsubsection{Reduction to biconnected factor-critical graphs}\label{sec:conj/ham/reduction}
We now show that in order to prove the conclusions in \cref{lem:h_bounds}, it suffices to prove simpler statements on a subset of unweighted graphs. We first introduce some helpful definitions.

\begin{definition}[Factor-critical graph]\label{def:factor_critical_graph}
    A graph $G(V,E)$ is \emph{factor-critical} if it is an unweighted graph such that the subgraph induced by $V \setminus \ofc{i}$ has a perfect matching for any $i \in V$.
\end{definition}

\begin{definition}[Biconnected graph]\label{def:biconnected_graph}
    A graph $G(V,E)$ is \emph{biconnected} if the subgraph induced by $V \setminus \ofc{i}$ is connected for any $i \in V$.
\end{definition}

Let $\mathcal{F}$ denote the set of biconnected factor-critical graphs. Note that due to \cref{def:factor_critical_graph}, graphs $F \in \mathcal{F}$ must be odd order. We then state the following lemma and defer the proof to \cref{apx:ham_reductions}.

\begin{restatable}{lemma}{LemHReductions}\label{lem:h_reductions}
\begin{align}
    \lmax\of{H^{QMC}(G)} \le m + \frac{n-1}{2}  &\implies \lmax\of{H^{QMC}(G)} \le W(G)+M(G)\,, \label{eq:qmc_reduction}\\
    \lmax\of{H^{EPR}(G)} \le m + \frac{n-1}{2}  &\implies \lmax\of{H^{EPR}(G)} \le  W(G)+M(G)\,, \label{eq:epr_reduction}\\
    \lmax\of{H^{XY}(G)} \le m + \frac{n-1}{4}  &\implies \lmax\of{H^{XY}(G)} \le  W(G)+\frac{M(G)}{2}\,, \label{eq:xy_max_reduction}\\
    \lmin\of{H^{XY}(G)} \ge -\frac{n-1}{4}  &\implies \lmin\of{H^{XY}(G)} \ge -\frac{M(G)}{2} \,\label{eq:xy_min_reduction},
\end{align}
    where the hypotheses must hold for all $G \in \mathcal{F}$ , and the conclusions are for any (possibly weighted) graph $G$.
\end{restatable}

In words, \cref{lem:h_reductions} states that if we can prove statements on the spectral radii of biconnected factor-critical graphs, we can prove statements on the spectral radii of \emph{all graphs}. We also may immediately identify useful corollary of the above reductions

\begin{corollary}
    For bipartite graphs, the conclusions in \cref{lem:h_bounds} hold.
\end{corollary}
\begin{proof}
    \cref{lem:h_reductions} asserts that if any of \cref{lem:h_bounds} are violated for some graph $G$, their corresponding hypotheses must also be violated by some biconnected factor-critical subgraph of $G$. As we later show in \cref{def:odd_ear_decomp}, a factor-critical subgraph must contain an odd-length cycle. However, a bipartite graph has no odd length cycles, so this is a contradiction.
\end{proof}

Like the case of Laplacians in \cref{lem:token_remove_triangles}, it suffices to consider triangle-free biconnected factor-critical graphs for QMC 
\begin{lemma}\label{lem:ham_remove_triangles}
    Suppose there exists some unweighted graph $G$ where $\lmax\of{H^{QMC}\of{G}} > m+M(G)$. Then there exists a triangle-free graph $G'$ with $m'$ edges such that $\lmax\of{H^{QMC}\of{G'}} > m'+M(G')$.
\end{lemma}
\begin{proof}
    If $G$ is triangle-free, we are done. Otherwise, decompose $G$ into a graph $H$ and the triangle $\Delta$. Then by the triangle inequality
    \begin{align*}
        \lmax\of{H^{QMC}\of{H}} &\ge \lmax\of{H^{QMC}\of{G}} - \lmax\of{H^{QMC}\of{\Delta}}\\
        &> m + M(G) - \lmax\of{H^{QMC}\of{\Delta}}.
    \end{align*}

We can check $\lmax\of{H^{QMC}\of{\Delta}} \le 3$, and we know $M(G) \ge M(H)$. Thus, $\lmax\of{H^{QMC}\of{H}} > m_H+M(H)$. We may repeat this argument now starting with $H$ until we have removed all triangles. A similar argument also works for $\lmax\of{L}$ to remove triangles. 
\end{proof} 

We present additional numerics, which in  particular test spectral bounds with respect to the maximum cut in \cref{apx:numerics/cut_bounds}.

\section{Implications of conjectures on approximation ratios}\label{sec:imp}

We now discuss how our conjectures in \cref{sec:conj} imply improved approximation ratios for QMC, XY, and EPR. We first present a general way to efficiently augment convex relaxations of the local Hamiltonian problems (such as \cite{parekh2021a, king2023, huber2024}) with constraints from maximum matchings, such that the relaxations obey \cref{lem:h_bounds}. We then present approximation algorithms for EPR, XY, and QMC that prepare simple tensor products of $1$ and $2$-qubit states. The algorithms and analysis are similar to those presented in \cite{anshu2020}.  Finally, we show how augmenting the state-of-the-art algorithm of \cite{apte2025} with matching constraints leads to a new state-of-the-art approximation ratio.

\subsection{Matching-augmented convex relaxations}\label{sec:imp/matching_sos}
We show how a resolution of \cref{conj:L_max,conj:Q_max,conj:A_max,conj:A_min} can be used to augment \emph{any} polynomial-time-solvable convex relaxation for the corresponding problem so that the solutions of the augmented convex relaxations also obey the bounds of \cref{lem:h_bounds}. All we require is that the constraints we add are indeed valid for any quantum state, which follows from our conjectures. A generic approach is possible in our setting because the bounds of \cref{lem:h_bounds} are closely related to constraints defining the matching polytope of a graph. Although there are an exponential number of such constraints, we appeal to the ellipsoid method in conjunction with an existing separation oracle for the matching polytope to obtain a polynomial-time algorithm.  

Our discussion will apply directly for the QMC or EPR problem, and the case of the XY Hamiltonian follows similarly upon rescaling and shifting the problem. We assume an instance $G=(V,E,w)$ and a convex programming relaxation $R$ that can be solved in polynomial time. The latter means that for any $\varepsilon > 0$, in time polynomial in the size of the instance and $O(\log(1/\varepsilon))$, one can obtain a point $x$ such that Euclidean distance of $x$ to $R$ (i.e., the feasible region of $R$) is at most $\varepsilon$ and the objective value of $x$ differs from the optimal by at most $\varepsilon$. 

\paragraph{Separation oracles} The most natural case of $R$ being solvable in polynomial time is when $R$ is a polynomial sized convex program; however, some convex programs with a superpolynomial number of constraints can also be solved efficiently. One such case uses a polynomial-time (weak) separation oracle that enables the ellipsoid method to solve $R$ in polynomial time. A separation oracle is a polynomial-time algorithm that when asked if a candidate point $\tilde{x}$ is in $R$, either certifies that $\tilde{x} \in R$ or produces a constraint $f(x) \leq 0$ that is valid for all $x \in R$ but is violated by $\tilde{x}$ (i.e., $f(\tilde{x}) > 0$). If $R$ has a polynomial number of constraints, a separation oracle can just simply check each constraint. A nontrivial example is the semidefinite constraint $X \succeq 0$ on $X \in \mathbb{R}^{d \times d}$, which is equivalent to the infinite family of linear constraints
\begin{align*}
\{a^TXa \geq 0 \mid a \in \mathbb{R}^d\}.
\end{align*}
A separation oracle is given by computing the smallest eigenvalue and corresponding eigenvector $u$ of $X$: this either certifies that $X \succeq 0$ or returns the violated constraint $u^TXu < 0$. 

Exact separation algorithms for general convex programs are not always possible using finite precision, and the ellipsoid method is also able to use \emph{weak} separation oracles that run in polynomial time and allow some error to solve convex programs in polynomial time. Since our matching-augmented approach is a standard application of the ellipsoid method with a separation oracle, we omit the detailed analysis of the running time and weak oracle precision; see Chapter 4 of \cite{groetschel1993} for more details. 

We further assume that the relaxation $R$ includes variables $\of{g_{ij} \in \ofb{0,2}}_{ij\in E}$ corresponding to the energy earned on each edge. Although the objective only depends on $g$, the relaxation $R$ may include polynomially many additional variables (such as the moment variables of fixed levels of the NPA/moment-SoS hierarchy~\cite{navascues2008}), which generally strengthens the relaxation. We give a polynomial-time separation oracle for $R$ augmented with the constraints 
\begin{align*}
\mathcal{C} \defeq \ofc{\sum_{(i,j) \in E} a_{ij} g_{ij} \leq M_a(G) + \sum_{(i,j) \in E} a_{ij} \;\middle|\; a \in \mathbb{R}_+^{|E|} },
\end{align*}
where $M_a(G)$ is the maximum weight of a matching in $G$ with weights $a$. Under \cref{conj:L_max,conj:Q_max,conj:A_max,conj:A_min}, all constraints in $\mathcal{C}$ are satisfied by
\begin{align*}
    g_{ij} \defeq \bra{\psi} h_{ij} \ket{\psi}
\end{align*}
arising from a quantum state $\ket{\psi}$, since by changing the weights so that $G \defeq (V,E,a)$,
\begin{align*}
    \sum_{(i,j) \in E} a_{ij} g_{ij} \leq \lmax(H(G)) \leq M_a(G) + \sum_{(i,j) \in E} a_{ij}
\end{align*}
by \cref{lem:h_bounds}. This means that augmenting $R$ with $\mathcal{C}$ still results in a relaxation of $\lmax(H(G))$.

\begin{algorithm}
\caption{Separation oracle for a matching-augmented convex relaxation}
\label{alg:sep_oracle}
\begin{algorithmic}[1]
\Require convex relaxation $R$ with objective variables $(g_{ij})_{ij \in E}$ and weak separation oracle $O_R$
\Input candidate point $(\tilde{g},\tilde{f})$ with $\tilde{g} = (\tilde{g}_{ij} )_{ij \in E}$ and $\tilde{f}$ corresponding to auxiliary variables of $R$
\State Query $O_R$ on $(\tilde{g},\tilde{f})$ \label{line:query_O_R}
\If{$O_R$ returns a violated constraint $c$}
    \State \Return $c$
\EndIf \label{line:return_violated_c}
\State Set $F \defeq \{(i,j) \in E : \tilde{g}_{ij} > 1\}$
\State Set $z_{ij} \defeq \tilde{g}_{ij}-1$ for all $(i,j) \in F$ \label{line:def_z}
\State Query the Padberg-Rao~\cite{padberg1982, groetschel1993} separation oracle, $O_{PR}$, for the matching polytope of the graph $H \defeq (V,F)$ on the candidate point $z$
\If{$O_{PR}$ returns a violated constraint $\sum_{(i,j) \in F(S)} z_{ij} > (|S|-1)/2$, for some odd $S \subseteq V$} \label{line:odd_set}
    \State Find $M$, the maximum size of a matching in the induced subgraph $H(S) \defeq (S, F(S))$ \label{line:matching}
    \State \Return the constraint $\sum_{(i,j) \in F(S)} g_{ij} \leq M + |F(S)|$
\ElsIf{$O_{PR}$ returns a violated constraint $\sum_{j \in N_H(i)} z_{ij} > 1$, for some $i \in V$} \label{line:star}
    \State \Return the constraint $\sum_{j \in N_H(i)} g_{ij} \leq 1 + |N_H(i)|$
\Else
    \State \Return ``$(\tilde{g}, \tilde{f})$ is feasible''
\EndIf
\end{algorithmic}
\end{algorithm}

\begin{remark} Using the Padberg-Rao separation oracle in \cref{alg:sep_oracle} allows us to enforce all the constraints of $\mathcal{C}$; however, recent SDP-based approximation algorithms~\cite{lee2024,apte2025} only need the single constraint $\sum_{(i,j) \in E} w_{ij} g_{ij} \leq M(G)$, which is implied by $\sum_{(i,j) \in F} w_{ij} z_{ij} \leq M(H)$. This constraint can be checked by find a maximum weight matching in $H$ and does not require Padberg-Rao, leading to a faster separation oracle.
\end{remark}

\begin{theorem}\cref{alg:sep_oracle} is a polynomial-time weak separation oracle for the convex relaxation $R$ augmented with the constraints $\mathcal{C}$.
\end{theorem}
\begin{proof}
We will let $O$ be the separation oracle described by \cref{alg:sep_oracle}. We first check that $O$ returns precisely violated constraints of $R$ or $\mathcal{C}$. \crefrange{line:query_O_R}{line:return_violated_c} ensure any violated constraint of $R$ is found. 

The Padberg-Rao separation oracle\footnote{The oracle is technically for the $b$-matching problem, and setting $b=1$ recovers the matching case.} checks for violations among the following full description of the matching polytope~\cite{padberg1982,edmonds1965} on $H$.
\begin{align}
      &\sum_{j\in N_H(i)} z_{ij} \le 1 &&\forall\,i\in V, \label{eq:lp_PR_star_constraint} \\
      &\!\!\! \sum_{(i,j)\in F(S)} z_{ij} \le \frac{|S|-1}{2} &&\forall\,S\subseteq V\!: |S| \; \text{odd}, \label{eq:lp_PR_odd_constraint} \\
      & z_{ij} \ge 0 &&\forall\,(i,j)\in F. \label{eq:lp_PR_non-negative_constraint}
\end{align}
For the constraints in $\mathcal{C}$, we first assume that the constraints of $R$ imply $g_{ij} \in [0,2]$ for all $(i,j) \in E$; if not, these can be explicitly added. In this case we have $z_{ij} \in [0,1]$ on \cref{line:def_z} so that \cref{eq:lp_PR_non-negative_constraint} is satisfied. If $O_{PR}$ returns a constraint of \cref{eq:lp_PR_odd_constraint} on \cref{line:odd_set}, then note that the constraint returned by $O$ in this case is in $\mathcal{C}$ by setting
\begin{align*}
    a_{ij} \defeq \begin{cases}
    1 & \text{if } (i,j) \in F(S), \\
    0 & \text{if } (i,j) \in E \setminus F(S),
\end{cases}
\end{align*}
with $M_a(G) = M$. Furthermore this constraint is violated by $\tilde{g}$ since
\begin{align*}
    \of{\sum_{(i,j) \in F(S)} \tilde{g}_{ij}}-|F(S)| = \sum_{(i,j) \in F(S)} z_{ij} > \frac{|S|-1}{2} \geq M,
\end{align*}
where the last inequality follows since $|S|$ is odd, so the maximum size matching in $F(S)$ cannot be a perfect matching and must leave at least one vertex uncovered.

If $O_{PR}$ returns a constraint of \cref{eq:lp_PR_star_constraint} on \cref{line:star}, we can similarly define
\begin{align*}
    a_{ij} \defeq \begin{cases}
    1 & \text{if } (i,j) \in N_H(i), \\
    0 & \text{if } (i,j) \in E \setminus N_H(i),
\end{cases}
\end{align*}
to see that the constraint returned by $O$ is in $\mathcal{C}$. This corresponds to the monogamy of entanglement inequality known as the star bound and is violated since
\begin{align*}
    \of{\sum_{j \in N_H(i)} \tilde{g}_{ij}}-|N_H(i)| = \sum_{j \in N_H(i)} z_{ij} > 1.
\end{align*}

Finally, if $O$ declares that $(\tilde{g},\tilde{f})$ is feasible, then all constraints of $R$ must be satisfied by \crefrange{line:query_O_R}{line:return_violated_c}. Since $O_{PR}$ ensures no matching polytope constraints can be violated, we have that $z$ is a convex combination of matchings $\{M^l\}_l$ in $H$:
\begin{align*}
z &= \sum_l \mu_l M^l \quad\text{with}\quad \sum_l \mu_l = 1 \quad\text{and}\quad \mu_l\geq 0,\,\forall l.
\end{align*}
This implies that for any $(a_{ij} \geq 0)_{ij \in E}$,
\begin{align*}
    \of{\sum_{(i,j) \in E} a_{ij}\tilde{g}_{ij}} - \sum_{(i,j) \in E} a_{ij} &= \sum_{(i,j) \in E} a_{ij} (\tilde{g}_{ij}-1)\\
    &\leq \sum_{(i,j) \in F} a_{ij} z_{ij}\\
    &= \sum_{l} \mu_l \sum_{(i,j) \in F} a_{ij} M^l_{ij}\\
    &\leq \max_l \sum_{(i,j) \in F} a_{ij} M^l_{ij}\\
    &\leq M_a(H)\\ 
    &\leq M_a(G),
\end{align*}
where the first inequality follows since $a_{ij} \geq 0$ and $\tilde{g}_{ij}-1 < 0$ for all $(i,j) \in E \setminus F$, and the last inequality follows since $H$ is a subgraph of $G$. Thus all constraints in $\mathcal{C}$ must be satisfied by $\tilde{g}$.

Our oracle $O$ runs in polynomial time since $O_{R}$ and $O_{PR}$ do. The only other substantial computation performed is computing a maximum size matching on \cref{line:matching}, which runs in polynomial time. As previously mentioned, weak separation oracles have precision requirements in addition to other mild requirements necessary for the ellipsoid method. We assume that $O_R$ satisfies these requirements since $R$ is solvable in polynomial time. The oracle $O_{PR}$ is an exact separation oracle and introduces no extra requirements.
\end{proof}

\subsection{EPR}\label{sec:imp/epr}
Consider the following algorithm, based on \cite[Section 5.1]{king2023} and \cite[Algorithm 1]{apte2025}

\begin{algorithm}
    \caption{EPR state preparation}
    \label{alg:epr_alg}
    Let $M$ be the set of edges in the maximum weight matching of $G$. Then, output the state
    \vspace{-.5em}
    \begin{align}\label{eq:epr_chi}
        \ket{\chi} = \prod_{(i,j) \in M} e^{i \theta_{ij} P_i P_j} \ket{0}^{\otimes n},
    \end{align}
    where
    \begin{align}\label{eq:epr_int_match_theta_assignment}
        P_i = \frac{X_i - Y_i}{\sqrt{2}}, \quad\quad \theta_{ij} = \frac{1}{2}\, \sin^{-1}\of{\frac{\sqrt{5}-1}{2}}.
    \end{align}
\end{algorithm}

We then show 

\begin{restatable}{corollary}{EPRApprox}\label{cor:epr_approx}
    Under \cref{lem:h_bounds}, \cref{alg:epr_alg} achieves an $\frac{1+\sqrt{5}}{4}\approx 0.809$-approximation ratio for EPR. This algorithm prepares a tensor product of $1$ and $2$-qubit states and is efficient.
\end{restatable}
The proof is deferred to \cref{apx:approx/epr}.

\subsection{XY}\label{sec:imp/xy}

Consider the following algorithm 

\begin{algorithm}
    \caption{XY state preparation}
    \label{alg:xy_alg}
    For any graph $G$, let $z_C$ be the bitstring corresponding to the maximum cut of $G$. Let $M$ be the set of edges in the maximum weight matching of $G$. Then, prepare the states
    \begin{enumerate}[label=\alph*)]
        \item $\rho_C$: the state corresponding to $z_C$ in the Pauli-$X$ basis.
        \item $\rho_M$: the state given by
        \begin{align}\label{eq:qmc_rho}
        \rho_M = \bigotimes_{(i,j) \in M} \ket{\psi^-}_{ij}\bra{\psi^-}_{ij} \bigotimes_{i \notin M} \rho_{mix},
    \end{align}
    where $\rho_{mix}$ is the maximally mixed state. Then, output the state corresponding to the larger of $\tr\ofb{H^{XY}(G) \rho_C}$ and $\tr\ofb{H^{XY}(G) \rho_M}$.
    \end{enumerate}
\end{algorithm}

It is clear that this algorithm is not efficient, as finding a maximum cut is $\NP$-hard. However, it does prepare a tensor product of $1$ and $2$-qubit states with the following approximation guarantee.

\begin{restatable}{corollary}{XYApprox}\label{cor:xy_approx}
    Under \cref{lem:h_bounds}, \cref{alg:xy_alg} obtains a $5/7 \approx 0.714$-approximation for XY. This algorithm prepares a tensor product of $1$ and $2$-qubit states, but is not efficient.
\end{restatable}

The proof is deferred to \cref{apx:approx/xy}. For an efficient algorithm, we have from \cite[Section 6]{briet2010} that there is an efficient algorithm that returns a randomized product state $\rho_{C'}$ with average energy at least $0.9349$ of the optimal product state (corresponding to $\nu(2)$ in their work). As $\rho_C$ is a product state, this algorithm obtains energy at least $0.9349 \,C$. Thus, we may replace $\rho_C$ with $\rho_{C'}$ to find

\begin{restatable}{corollary}{XYApproxEfficient}\label{cor:xy_approx_efficient}
    Under \cref{lem:h_bounds}, \cref{alg:xy_alg} with $\rho_C$ replaced by $\rho_C'$ obtains a $0.674$-approximation in expectation for XY. This algorithm prepares a tensor product of $1$ and $2$-qubit states and is efficient.
\end{restatable}

The proof is deferred to \cref{apx:approx/xy}. 

\subsection{QMC}\label{sec:imp/qmc}
We finally present implied approximation ratios for QMC.

\begin{restatable}{corollary}{QMCMATCHApprox}\label{cor:qmc_match_approx}
    Under \cref{lem:h_bounds}, \cref{alg:xy_alg} obtains a $5/8=0.625$-approximation for QMC. This algorithm prepares a tensor product of $1$ and $2$-qubit states but is not efficient.
\end{restatable}

\begin{restatable}{corollary}{QMCMATCHApproxEfficient}\label{cor:qmc_match_approx_efficient}
    Under \cref{lem:h_bounds}, \cref{alg:xy_alg} with $\rho_C$ replaced by $\rho_C'$  obtains a $0.604$-approximation for QMC. This algorithm prepares a tensor product of $1$ and $2$-qubit states and is efficient.
\end{restatable}
Both proofs are deferred to \cref{apx:approx/qmc}. Our conjectured bounds also would improve the state-of-the-art approximation algorithm for QMC, given in \cite{apte2025}. 

\begin{restatable}{corollary}{QMCPMATCHApprox}\label{cor:qmc_pmatch_approx}
    Under \cref{lem:h_bounds}, \cite[Algorithm 3]{apte2025} augmented with the matching-based separation oracle as
   in \cref{alg:sep_oracle} obtains a $0.614$-approximation for QMC. This algorithm prepares a tensor product of $1$ and $2$-qubit states and is efficient.
\end{restatable}

The proof is deferred to \cref{apx:approx/qmc}

\subsection{Additional Numerics}\label{sec:imp/additional_numerics}
Finally, besides our conjectured upper bounds, we also numerically evaluate the approximation ratios achieved by \cref{alg:xy_alg} with $\rho_C$ or $\rho'_C$ as defined in \cref{sec:imp/xy} and \cref{sec:imp/qmc}. We evaluate this ratio by finding the minimum approximation over the graphs in our dataset. We present these results in \cref{tab:numeric_approx}.

\begin{table}[H]
    \centering
    \begin{tabular}{c||c|c}
          & \cref{alg:xy_alg} with $\rho_C$ & \cref{alg:xy_alg} with $\rho'_C$\\
        \hline
        \hline
         QMC &  $0.747$ & $0.734$  \\
        \hline
         XY &  $0.811$ & $0.766$ \\
    \end{tabular}
    \caption{Approximation ratios achieved by \cref{alg:xy_alg} with $\rho_C$ or $\rho'_C$ as defined in \cref{sec:imp/xy} and \cref{sec:imp/qmc} for XY and QMC, respectively. Note that the algorithms with $\rho_C$ are not efficient, but rather demonstrate the approximations achievable by a tensor product of $1$ and $2$-qubit states.}
    \label{tab:numeric_approx}
\end{table}

We do not conjecture particular improved approximation ratios based on results in \cref{tab:numeric_approx}, but rather provide this as evidence that the algorithms may perform better than the analyzed worst-case guarantees. We do not present this data for EPR as it was already shown in \cite{apte2025} that a $\frac{1+\sqrt{5}}{4}$-approximation is tight.

\section{Discussion}
We show that quantum MaxCut and related Hamiltonians defined on graphs can be studied purely in terms of spectral properties of \emph{token graphs}. This layer of abstraction allows for the study of the Hamiltonians without any prior exposure to physics or quantum computing. Based on numerical evidence, we pose conjectures bounding the spectral radii of token graphs. These bounds are stronger than previous results, and as such we believe our conjectures to be of independent interest to the graph theory community. Thus, we invite any, regardless of background, to prove or refute the conjectures. We show that \cref{conj:L_max,conj:Q_max,conj:A_max,conj:A_min}, which hold for unweighted graphs, imply novel combinatorial bounds on the spectral of \emph{weighted} token graphs. We are not aware of previous studies on weighted token graphs, and these bounds are tighter than known bounds even for unweighted graphs.

Interestingly, it is known that QMC is $\QMA$-complete \cite{piddock2015, cubitt2016}. As such, the presentation of QMC in terms of token graphs allows for a fully classical (albeit exponentially sized) characterization of a $\QMA$-complete problem solely in terms of graph-theoretic objects. Similar results hold for XY ($\QMA$-complete) and EPR ($\StoqMA$) \cite{piddock2015}. 

For a reader familiar with statistical mechanics and/or quantum computation, we show that analyzing specific $2$-local Hamiltonians from a graph-theoretic perspective may yield improved upper bounds and approximation ratios. In particular our conjectures imply a tighter \emph{analysis} of existing algorithms, leading to higher approximation ratios. Furthermore, we find numerical evidence that these algorithms may obtain approximations even better than those implied by our conjectures. As such, we emphasize the importance of finding tighter upper bounds for these Hamiltonians. We collect and summarize open problems below
\begin{itemize}
    \item Prove \cref{conj:L_max,conj:Q_max,conj:A_max,conj:A_min,conj:A_monotonic,conj:Q_monotonic}.
    \item Tighten the analysis of existing algorithms for QMC and XY to achieve better approximation ratios, potentially closer to the values presented in \cref{tab:numeric_approx}.
    \item Determine if there exist any classical graph invariants that yield the minimum value $k$ for which $\lmax\of{L\of{F_k\of{G}}}=\lmax\of{H^{QMC}\of{G}}$? For example, we check if $k=M(G)$ or the number of vertices in a maximum bipartite subgraph of $G$, but neither of these are true. Similar questions can be asked for EPR and $Q(F_k(G))$ or XY and $-A(F_k(G))$.
    \item Design new algorithms for QMC, XY, and EPR based on properties of eigenvectors of token graphs, such as the algorithm for EPR presented in \cref{apx:other_app/algo}.
\end{itemize}

\section*{Acknowledgements}
The authors thank Eunou Lee, Kunal Marwaha, and Willers Yang for helpful discussions.
This work is supported by a collaboration between the US DOE and other Agencies.
The work of A.A is supported by the Data Science Institute at the University of Chicago.  
O.P. acknowledges that this material is based upon work supported by the U.S. Department of Energy, Office of Science, Accelerated Research in Quantum Computing, Fundamental Algorithmic Research toward Quantum Utility (FAR-Qu). 
J.S. acknowledges that this material is based upon work supported by the National Science Foundation Graduate Research Fellowship under Grant No.\ 2140001. J.S acknowledges that this work is funded in part by the STAQ project under award NSF Phy-232580; in part by the US Department of Energy Office of Advanced Scientific Computing Research, Accelerated Research for Quantum Computing Program.

This article has been authored by an employee of National Technology \& Engineering Solutions of Sandia, LLC under Contract No.\ DE-NA0003525 with the U.S. Department of Energy (DOE). The employee owns all right, title and interest in and to the article and is solely responsible for its contents. The United States Government retains and the publisher, by accepting the article for publication, acknowledges that the United States Government retains a non-exclusive, paid-up, irrevocable, world-wide license to publish or reproduce the published form of this article or allow others to do so, for United States Government purposes. The DOE will provide public access to these results of federally sponsored research in accordance with the DOE Public Access Plan \url{https://www.energy.gov/downloads/doe-public-access-plan}.

\bibliography{refs}

\clearpage
\newpage

\appendix

\section{Other applications of token graph equivalence}\label{apx:other_app}
We now list  applications of the connection between Hamiltonians and token graphs that were not included in \cref{sec:equivalence/app} and provide omitted proofs

\subsection{Spectrum containment}\label{apx:other_app/spectrum_containment}
\LemLaplacianSpectrumContainment*
\begin{proof}
    We extend the local approach of \cite[Theorem 4.1]{dalfo2021} to weighted graphs. Let $F_k(G)$ and $F_{k+1}(G)$ be shorthand for $F_{k}(G)$ and $F_{k+1}(G)$. Let $L_{k+1}$ and $L_k$ be shorthand for $L(F_{k+1})$ and $L(F_k)$. Let $V_k$, $V_{k+1}$ denote $\binom{\ofb{n}}{k}$ and $\binom{\ofb{n}}{k+1}$, the vertex sets of $F_k$ and $F_{k+1}$. We will index vertices in $V_k$ with $X$ and $Y$ and vertices in $V_{k+1}$ with $S$ and $T$. Let $v_k$ be an eigenvector of $F_k(G)$ with corresponding eigenvalue $\lambda$. Let $J_S \defeq \ofc{X \in V_k: X \subset S}$. Then define
    \begin{align*}
        v_{k+1}(S) \defeq \sum_{X \in J_S} v_k(X), \quad\quad \forall S \in F_{k+1}. 
    \end{align*}
    We now show that $v_{k+1}$ is an eigenvector $F_{k+1}(G)$ with eigenvalue $\lambda$ as well. Given a vertex $S$ of $F_{k+1}$, let $L_{k+1}(S)$ denote the row of $L_{k+1}$ corresponding to $S$. Then,
    \begin{align*}
        L_{k+1}(S)v_{k+1} &= \sum_{T \in N(S)} w_{ST}\of{v_{k+1}(S)-v_{k+1}(T)},\\
        &= \sum_{T \in N(S)} w_{ST}\of{\sum_{X\in J_S}v_{k+1}(X)-\sum_{Y\in J_T}v_{k+1}(Y)},\\
        &= \sum_{X \in J_A} \sum_{Y \in N(X)}w_{XY}\of{v_k(X)-v_k(Y)} ,\\
        &= \sum_{X \in J_S} \lambda\, v_k(X) = \lambda\, v_{k+1}(S).
    \end{align*}
    
    The reason of the third line is the following. First, notice that for each $T \in N(S)$, we have that $S = Z \cup \ofc{s}$ and $T = Z \cup \ofc{t}$, where $Z=S\cap T$ and $(s,t)\in G$. Now, if either $X,Y \subset S$ or $X,Y \subset T$, both terms $v_k(X)$ and $v_k(Y)$ appear in both second and third sums of the second line, canceling out. Otherwise, for each $X \in J_S$ such that $x \in X$, there is one $Y \in J_T$ such that $y \in Y$ we have $X \triangle Y = S \triangle T$ and $Y \in N(X)$, with weight $W_{XY}$. 
\end{proof}

\subsection{Monogamy of Entanglement for XY on a star}
\StarAMax*
\begin{proof}
    It is shown in \cite{dalfo2021} that $F_k(S_m) \cong J(m,k-1,k)$, where $J(m,k-1,k)$ denotes the \emph{doubled} Johnson graph, which is a bipartite graph consisting over vertices $U \cup V$, where $P = \binom{\ofb{m}}{k-1}$, $Q= \binom{\ofb{m}}{k}$ and two vertices $p$ and $q$ are adjacent if and only if $p \subset q$ or $q \subset p$. We may write the adjacency matrix of $J(m,k-1,k)$ as follows 
    \begin{align}\label{eq:A_star_block}
    A\of{F_k\of{S_{m}}}=\begin{bmatrix}
    \mathbf{0_{\binom{m}{k-1}}} & B \\
    \vspace{-8px}\\
    B^T & \mathbf{0}_{\binom{m}{k}}
    \end{bmatrix},
    \end{align}
    where vertices are first ordered by $P$ and then by $Q$, $B$ is a binary $\binom{m+1}{k} \times \binom{m+1}{k-1}$ matrix with $(m+1-k)$ nonzero entries per row, and $\mathbf{0}_d$ is the square all-zero matrix with dimension $d$. We may bound the eigenvalues of $A$ using a similarity transform
    \begin{align*}
        &\lmax\of{A\of{F_k\of{S_{m}}}} = \lmax\of{A'\of{F_k\of{S_{m}}}}\\
        &A' \defeq D\of{F_k\of{S_{m}}}^{1/2} A\of{F_k\of{S_{m}}}\, D\of{F_k\of{S_{m}}}^{-1/2}.
    \end{align*}
    We then invoke Gershegorin's circle theorem on the right hand side. The similarity transform multiplies nonzero entries in the upper right block of \cref{eq:A_star_block} by $\sqrt{\frac{k}{m+1-k}}$ and the lower right block by $\sqrt{\frac{m+1-k}{k}}$. Thus, taking the sums of absolute values of rows yields
    \begin{align*}
        \sum_{i} \of{A'_{ij}} =
        \begin{cases}
            \of{m+1-k}\sqrt{\frac{k}{m+1-k}} = \sqrt{k(m+1-k)}, & i \in P \\
            k\sqrt{k(m+1-k)} = \sqrt{k(m+1-k)}, & i \in Q \, .
        \end{cases}
    \end{align*}
    Maximizing over $i$ then yields $\lmax\of{A'\of{F_k\of{S_{m}}}}\le\sqrt{k(m+1-k)}$. As all the rows are equal, the bound is tight.
\end{proof}

\subsection{Complete graphs}\label{apx:other_app/complete}
Using the equivalence to token graphs, we now derive the eigenspectrum of QMC, XY, and EPR on unweighted complete graphs. This analysis has already been done (e.g. \cite{osborne2006}), but we repeat in our notation for completeness. The token graphs $F_k(G)$ of the complete graph $K_n$ are known to be isomorphic to the Johnson graph $J(n,k)$ (see, for example \cite{fabila-monroy2012}). The Johnson graph is $k(n-k)$ regular and the adjacency matrix spectrum is fully characterized 

\begin{lemma}\label{lem:eigenvalues_johnson}
    For any $n \ge 2$, $1 \le k \le \floor{n/2}$
\begin{align*}
    \eigs\of{A\of{F_k\of{K_n}}} &= \ofc{(k-i)(n-k-i)-i}_{\,0 \le i \le k}\,, \\
    \eigs\of{L\of{F_k\of{K_n}}} &=  k(n-k) - \eigs\of{A\of{F_k\of{K_n}}}, \\
    \eigs\of{Q\of{F_k\of{K_n}}} &= k(n-k) + \eigs\of{A\of{F_k\of{K_n}}}.
\end{align*}
\end{lemma}

In particular we find that $\lmax\of{A\of{F_k\of{K_n}}}=k(n-k)$, where we $i=0$ and $\lmax\of{-A\of{F_k\of{K_n}}}=k$ where $i=k$. Combining \cref{lem:eigenvalues_johnson} with \cref{fact:xy_A_equivalence}, \cref{fact:qmc_L_equivalence}, and \cref{fact:epr_Q_equivalence} we may optimize over $k$ to obtain

\begin{lemma}\label{lem:H_eigenvalues_complete}
    For any $n \ge 2$,
    \begin{align*}
    \lmax\of{H^{XY}(K_n)} &= 
    \begin{cases}
        \frac{n^2+n}{4}, & n \,\text{even}, \\
        \frac{n^2+n-2}{4}, & n \,\text{odd}.      
    \end{cases}, \\
    \lmax\of{H^{QMC}(K_n)} &=  
    \begin{cases}
        \frac{n^2+2n}{4}, & n \,\text{even}, \\
        \frac{n^2+2n-3}{4}, & n \,\text{odd}.
    \end{cases}, \\
    \lmax\of{H^{EPR}(K_n)} &= 
    \begin{cases}
        \frac{n^2}{2}, & n \,\text{even}, \\
        \frac{n^2-1}{2}, & n \,\text{odd}.
    \end{cases}.
\end{align*}
\end{lemma}

\subsection{Maximum energy states for EPR have positive amplitudes}\label{apx:other_app/epr_positive}

The connection to token graphs yields the following observation for EPR, which was observed for general stoquastic Hamiltonians in \cite[Lemma 3]{albash2018}
\begin{lemma}\label{lem:epr_non_negative}
    For any connected graph $G$, there exists a quantum state $\ket{\chi}$ satisfying 
    \begin{align*}
        &\ofk{\chi\middle|H^{EPR}(G)\middle|\chi} = \lmax\of{H^{EPR}\of{G}}, \\
        &\ofk{z \middle| \chi} \in \mathbb{R}_+, \quad \forall \ket{z} \in \{ \ket{+i}, \ket{-i} \}^{\otimes n},
    \end{align*}
    where $\ket{+i}$, $\ket{-i}$ are the eigenstates of the Pauli $Y$ operator. In words, $H^{EPR}(G)$ admits an optimal state with real, non-negative coefficients in the Pauli-$Y$ basis
\end{lemma}
\begin{proof}
    By the analysis in \cref{sec:equivalence/epr}, EPR in the Pauli-Y basis is equivalent to a direct sum of signless Laplacians of token graphs. The signless Laplacian matrix of a connected graph consists of all non-negative entries, and can easily be seen to be irreducible. As such, we may invoke the Perron-Frobenius theorem to show that $Q\of{F_k\of{G}}$ always has an eigenvector $v$ corresponding to eigenvalue $\lmax\of{Q\of{F_k\of{G}}}$ such that all components of $v$ are positive. By \cref{fact:epr_Q_equivalence}, this means that there exists an optimal state for $H^{EPR}(G)$ that admits a decomposition where all states $\ket{z}$ corresponding to bitstrings in the Pauli-Y basis with Hamming weight $k$ have positive amplitude and $0$ else, for some value of $k$ in $1\le k \le \floor{n/2}$. 
\end{proof}

\subsection{New algorithmic primitives for EPR}\label{apx:other_app/algo}

For the following observation, we must compute the weight of token graphs
\begin{lemma}\label{lem:token_graph_weight}
    For any graph $G$ and any $1\le k < n$
    \begin{align*}
        W(F_k(G)) = \binom{n}{k} \frac{k(n-k)}{n(n-1)} W(G).
    \end{align*}
\end{lemma}
\begin{proof}
    By definition
    \begin{align*}
        W\of{F_k(G)}
        &=\sum_{\substack{A,B\subseteq [n]\\|A|=|B|=k\\A\triangle B=\ofc{a,b}}}
           w_{ab}\,.
    \end{align*}
    Fix an edge $\of{a,b}\in E$.  To form a pair $\of{A,B}$ with $A\triangle B=\ofc{a,b}$, one must choose $\of{k-1}$ of the remaining $\of{n-2}$ vertices to lie in $A\setminus\{a\}$ (equivalently in $B\setminus\{b\}$).  There are $\binom{n-2}{k-1}$ ways to do this. Since $F_k(G)$ is simple and undirected, each such $\of{A,B}$ counts exactly one edge of weight $w_{ab}$.  Summing over all original edges,
    \begin{align*}
        W\of{F_k(G)}
        =\sum_{\ofc{a,b}\in E}
            \binom{n-2}{k-1}\,w_{ab}
        =\binom{n-2}{k-1}\,\sum_{\{a,b\}\in E}w_{ab}
        =\binom{n-2}{k-1}\;W(G).
    \end{align*}
    The proof then follows from the standard binomial identity
    \begin{align*}
        \binom{n-2}{k-1}=\binom{n}{k}\;\frac{k\,(n-k)}{n\,(n-1)}.
    \end{align*}
\end{proof}

Now consider the following algorithmic lower bound for the maximum energy of the signless Laplacian.

\begin{algorithm}
    \caption{Uniform-vector lower bound for the signless Laplacian}
    \label{alg:uniform_Q}
    Given $G(V,E,w)$ and a fixed $1\le k \le \floor{n/2}$ output the normalized $1$ vector
    \begin{align}\label{eq:v_n_k_vector}
        v_{n,k}= \binom{n}{k}^{-\frac{1}{2}}\,\vec{1}.
    \end{align}
\end{algorithm}

\begin{lemma}\label{lem:uniform_Q_alg_approx}
    The vector $v_{n,k}$ in \cref{eq:v_n_k_vector} outputted by \cref{alg:uniform_Q} is an eigenvector of $Q\of{F_k\of{G}}$ with eigenvalue 
    \begin{align*}
        4 W(G) \frac{k(n-k)}{n(n-1)}.
    \end{align*}
\end{lemma}
\begin{proof}
    Consider the matrix $Q\of{F_k\of{G}} + L\of{F_k\of{G}} = 2 \,D\of{F_k\of{G}}$. It is well known that the all-$1$ vector $v_{n,k}$ satisfies $ L\of{F_k\of{G}} \,v_{n,k} = 0$. Thus 
    \begin{align*}
        v_{n,k}^T \, Q\of{F_k\of{G}} \, v_{n,k} &= v_{n,k}^T  \of{Q\of{F_k\of{G}} + L\of{F_k\of{G}}} \, v_{n,k}\,, \\
        &= 2 \, v_{n,k}^T \,D\of{F_k\of{G}}  \, v_{n,k}\,, \\
        &= \frac{2}{\binom{n}{k}}\sum_{A \in F_k(G)} \of{\sum_{B \in N(A)} w_{AB}}\,, \\
        &= \frac{4}{\binom{n}{k}} W(F_k(G))\,, \\
        &= 4 W(G)\frac{k(n-k)}{n(n-1)}\,,
    \end{align*}
    where in the third line we used the definition of a degree matrix, in the fourth line we used that the sum of (weighted) degrees in a graph is equal to twice the weight of the graph, and in the last line we used \cref{lem:token_graph_weight}.
\end{proof}

We now consider the following algorithm for EPR 

\begin{algorithm}
    \caption{Dicke-state preparation for EPR}
    \label{alg:dicke}
    Given $G(V,E,w)$ output the quantum Dicke state \cite{dicke1954}
    \begin{align*}
        \binom{n}{\floor{n/2}}^{-\frac{1}{2}}\!\!\sum_{z \in \binom{\ofb{n}}{\floor{n/2}}} \ket{z},
    \end{align*}
    where $\ket{z}$ encodes the bitstring $z \in \mathbb{R}^n$ in the Pauli $Y$-basis. This state is simply a uniform superposition over Hamming weight $\floor{n/2}$ states in the $Y$-basis. Furthermore, these states can be prepared efficiently (see for instance \cite{bartschi2022}).
\end{algorithm}

\begin{corollary}\label{cor:dicke_algo_epr}
    On any graph $G$, \cref{alg:dicke} achieves energy 
    \begin{align*}
        \ofk{D^n_k\middle|H^{EPR}(G)\middle|D^n_k} = 
        \begin{cases}
            W(G) \frac{n}{n-1} & n \text{ even}\,, \\
            W(G) \frac{n+1}{n} & n \text{ odd}\,. \\
        \end{cases}
    \end{align*}
\end{corollary}
\begin{proof}
    By \cref{fact:epr_Q_equivalence} we have that $H^{EPR}(G)$ is equivalent under a unitary transformation (in particular, the transformation by $U$ in \cref{eq:U_transform}) to a direct sum over the signless Laplacian matrices $Q(F_k(G))$. In particular, we can see in \cref{eq:U_transform} that $U$ transforms basis states of $Y$ to basis states of $Z$. Thus, $U$ transforms $\ket{D^n_k}$ to the Dicke state in the $Z$-basis. In the computational basis, it is then clear that $\ket{D^n_k}$ maps to the vector $v_{n,k}$ from \cref{eq:v_n_k_vector}. This vector was shown in \cref{lem:uniform_Q_alg_approx} to achieve energy $4 W(G)\frac{k(n-k)}{n(n-1)}$. Maximizing this quantity over $k$ yields the proof at $k=\floor{n/2}$.
\end{proof}

\section{Proof of \texorpdfstring{\cref{lem:token_reductions}}{token graph reductions}}\label{apx:token_reductions}

We first prove a useful lemma

\begin{lemma}\label{lem:match_k_lp}
    Let $\mathrm{MATCH}_k$ denote the following linear program defined on some graph $G(V,E,w)$  
    \begin{align}
      \max\quad & \sum_{(i,j)\in E} w_{ij}\,z_{ij} \nonumber \\
      \text{s.t.}\quad
      & \sum_{j\in N(i)} z_{ij} \le 1 &&\forall\,i\in V, \label{eq:lp_k_star_constraint} \\
      &\!\!\! \sum_{(i,j)\in E(S)} z_{ij} \le \frac{|S|-1}{2} &&\forall\,S\subseteq V\!: |S| \; \text{odd}, \label{eq:lp_k_odd_constraint} \\
      & z_{ij} \ge 0 &&\forall\,(i,j)\in E, \label{eq:lp_k_non-negative_constraint}\\
      & \sum_{(i,j)\in E} z_{ij} \le k \label{eq:lp_k_constraint}.
    \end{align}
     The optimal value of $\mathrm{MATCH}_k$ is $M_k(G)$.
\end{lemma}
\begin{proof}
    The proof follows from \cite[Theorem 42]{araoz1983}, which concerns $b$-matchings, a generalization of matchings. A $b$-matching is a multiset $B$ of edges $e\in E$ such that each vertex $v \in V$ is incident to at most $b_v$ edges in B, and where $B$ contains at most $u_e$ copies of each edge $e$. We let $b_v$ and $u_e$ be parameters of the $b$-matching. When all $b_v, u_e =1$, we recover the standard definition of a matching. Our result follows by taking all $b_v, u_e =1$, $p=0$, and $q=k$ in Theorem 42. We then get that the convex hull of matchings with at most $k$ edges is given by all $z$ satisfying \cref{eq:lp_k_star_constraint,eq:lp_k_non-negative_constraint,eq:lp_k_constraint} and
    \begin{align}
        \sum_{(i,j)\in E(T) \cup F} z_{ij} \le \frac{|T|+|F|-1}{2} &&\forall\,T\subseteq V, F \subseteq \delta(T) \!: |T|+|F| \; \text{odd}.  
    \label{eq:capacitated_odd_constraint}
    \end{align}
    Note that the above constraints are a generalization of \cref{eq:lp_k_odd_constraint} by setting $F = \emptyset$. To complete the proof, we claim that \cref{eq:lp_k_star_constraint,eq:lp_k_non-negative_constraint,eq:lp_k_constraint,eq:lp_k_odd_constraint} imply \cref{eq:capacitated_odd_constraint}. Thus, the feasible solution set of $\mathrm{MATCH}_k$ is the convex hull of matchings with at most $k$ edges. The optimal solution of the LP must lie at a vertex of the feasible set, and since these vertices are precisely given by matchings with at most $k$ edges, the optimal value $\mathrm{MATCH}_k$ corresponds to $M_k(G)$. It then suffices to prove the claim.
    
    Suppose $z$ satisfies \cref{eq:lp_k_star_constraint,eq:lp_k_non-negative_constraint,eq:lp_k_constraint,eq:lp_k_odd_constraint}. Let $N_F(T) = \{j \in V \setminus T \mid \exists \,i \in T \text{ with } (i,j) \in F\}$ be the set of vertices outside $T$ that neighbor a vertex in $T$ via an edge in $F$. Observe that $|N_F(T)| \leq |F|$. First consider the case when $|N_F(T)| < |F|$. Let $S = T \cup N_F(T)$. By summing \cref{eq:lp_k_star_constraint} over $i \in S$, we get
    \begin{align*}
        \frac{1}{2}\sum_{i\in S}\sum_{j \in N(i)}z_{ij}=\sum_{(i,j)\in E(S)} z_{ij} + \frac{1}{2} \sum_{(i,j)\in \delta(S)} z_{ij} \leq \frac{|S|}{2},
    \end{align*}
    where the middle expression comes from the fact that we double count edges in $E(S)$ and single count edges in $\delta(S)$. Since $z_{ij} \geq 0$, this implies
    \begin{align*}
        \sum_{(i,j)\in E(S)} z_{ij} \leq \frac{|S|}{2}.
    \end{align*}
    Since $E(T) \cup F \subseteq E(S)$, we get 
    \begin{align*}
        \sum_{(i,j) \in E(T)\cup F} z_{ij} = \sum_{(i,j)\in E(S)} z_{ij} \leq \frac{|S|}{2} \leq \frac{|T|+|F|-1}{2},
    \end{align*}
    as desired.

    The other case is when $|N_F(T)| = |F|$. For this we take $S = T \cup N_F(T)$ as above and observe that $|S| = |T| + |F|$ so that \cref{eq:capacitated_odd_constraint} follows directly from \cref{eq:lp_k_odd_constraint}.
\end{proof}

We now may prove \cref{lem:token_reductions}. We start with \cref{eq:l_max_reduction}, and we show how the rest of the statements follow similarly. We restate the portion of the lemma concerning the Laplacian below

\begin{claim}[Restatement of \cref{eq:l_max_reduction}]\label{claim:l_unweighted_reduction}
    If
    \begin{align}
        \lambda_{max}\of{L\of{F_k\of{G}}} \le m+k \label{eq:reduction_hypothesis},
    \end{align}
    for any unweighted graph $G$ and any $1\le k\le\floor{\frac{n}{2}}$. Then, 
    \begin{align}\label{eq:reduction_conclusion}
        \lambda_{max}\of{L\of{F_k\of{G}}} \le W(G)+M_k(G)\,,
    \end{align}
    for any (possibly weighted) graph $G$.
\end{claim}

\begin{proof}
    We prove this statement by contradiction. Suppose \cref{eq:reduction_hypothesis} is true and \cref{eq:reduction_conclusion} is false. Let $G$ be the minimal graph (i.e. the graph with smallest number of edges) for which \cref{eq:reduction_conclusion} is false. For some unit-norm eigenvector $v$ of $L(F_k(G))$ with eigenvalue $\lambda$, let $x_e \defeq v^T L(F_k(G_e))\,v$, where $G_e$ is the graph consisting of just the edge $e$. If there is some optimal eigenvector $v^*$ such that there exists an edge $e$ with $x^*_e \leq 1$, then we can find a smaller $G$ for which \cref{eq:reduction_conclusion} is false by removing $e$ as follows: let $G'$ be the graph with edge $e'$ removed from $G$. Then
    
    \begin{align*}
        \lambda_{max}\of{L(F_k(G))} &\geq \sum_{e \in E \neq e'} w_e x^*_e =  \lambda_{max}(G) - w_{e'} x^*_{e'} \\
        &> W(G) + M_k(G) - w_{e'}  \geq W(G') + M_k(G'),
    \end{align*}
    
    where we used $\lambda_{max}(L(F_k(G))) > W(G) + M_k(G)$,  $-x^*_e \geq -1$, and $w_e \geq 0$ for the strict inequality and $M_k(G) \geq M_k(G')$ and $W(G')=W(G)-w_e$ for the final inequality. This means $\lambda_{max}\of{L(F_k(G))} > W(G') + M_k(G')$, so $G$ was not a minimal graph for the hypothesis. Thus, we may assume that for every optimal solution $v^*$ of $G$, we have that $x^*_e \geq 1$ for all $e$. Now let $y^*_e \defeq x^*_e - 1$. Because $v^*$ is optimal, we know 
    \begin{align*}
        \lambda_{max}\of{L(F_k(G))} = \sum_e w_e x^*_e = \sum_e w_e (1+y^*_e) = W(G) + \sum_e w_e y^*_e.
    \end{align*} 
    By assumption then $\sum_e w_e y^*_e > M_k(G)$. Thus, if \cref{lem:match_k_lp} is true, then one of the constraints in the LP must be violated by the values $y$. 
    However
    \begin{itemize}
        \item \cref{eq:lp_k_star_constraint} is satisfied by \cref{lem:l_star_bound}.
        \item \cref{eq:lp_k_odd_constraint} is satisfied because the subgraph $G(S)$ has an odd number of nodes $S$. From \cref{eq:reduction_hypothesis} we have that $\sum_{e \in G(S)} x^*_e \leq W(S) + k'$ for some $1 \le k' \leq \floor{|S|/2}=\frac{S-1}{2}$. Thus $\sum_{e \in G(S)} y^*_e = \sum_{e \in G(S)} x^*_e - W(S) \le \frac{|S|-1}{2}$. 
        \item \cref{eq:lp_k_non-negative_constraint} is satisfied because $x^*_e \geq 1$.
        \item \cref{eq:lp_k_constraint} is satisfied by \cref{eq:reduction_hypothesis}, applied to the unweighted version of the graph $G$.
    \end{itemize}
     We thus have a contradiction.
\end{proof}

The proof for $\lmax\of{Q(F_k(G))}$ follows identically to that of \cref{claim:l_unweighted_reduction}. The proofs of $\lmax\of{A(F_k(G))}$ and $\lmin\of{A(F_k(G))}$ requires us use \cref{lem:star_A_max} and \cref{cor:star_A_min}, which in particular imply that the absolute value of the maximum and minimum eigenvalues of the adjacency matrices of the $k$\emph{-th} token graph of a star are less than $\frac{m+k}{2}$. 
Then, we simply repeat the proof of \cref{claim:l_unweighted_reduction}, using all the same definitions, albeit with respect to the Hamiltonians $2A(F_k(G))$ and $-2A(F_k(G))$. The monogamy of entanglement result on a star then shows that the maximum eigenvalue of these scaled matrices are $m+k$, analogously to the case of the Laplacian and signless Laplacian.

\section{Proof of \texorpdfstring{\cref{lem:h_reductions}}{Hamiltonian reductions}}\label{apx:ham_reductions}

We first use the following characterization of the matching polytope

\begin{lemma}[\cite{pulleyblank1974}]\label{lem:match_lp}
    Let $\mathrm{MATCH}$ denote the following linear program defined on some graph $G(V,E,w)$  
    \begin{align}
      \max\quad & \sum_{(i,j)\in E} w_{ij}\,z_{ij} \nonumber \\
      \text{s.t.}\quad
      & \sum_{j\in N(i)} z_{ij} \le 1 &&\forall\,i\in V, \label{eq:lp_star_constraint} \\
      &\!\!\! \sum_{(i,j)\in E(S)} z_{ij} \le \frac{|S|-1}{2} &&\forall\,S\subseteq V\!: G(S) \in \mathcal{F}, \label{eq:lp_2vcfc_constraint} \\
      & z_{ij} \ge 0 &&\forall\,(i,j)\in E \label{eq:lp_non-negative_constraint}.
    \end{align}
     The optimal value of $\mathrm{MATCH}$ is $M(G)$.
\end{lemma}

We now may prove \cref{lem:h_reductions}. We start with \cref{eq:qmc_reduction}, and we show how the rest of the statements follow similarly. We restate the portion of the lemma concerning QMC below

\begin{claim}[Restatement of \cref{eq:qmc_reduction}]\label{claim:qmc_reduction}
    If
    \begin{align}
        \lambda_{max}\of{H^{QMC}\of{G}} \le m+M(G) \label{eq:ham_reduction_hypothesis},
    \end{align}
    for any $G \in \mathcal{F}$, then 
    \begin{align}\label{eq:ham_reduction_conclusion}
        \lambda_{max}\of{H^{QMC}(G)} \le W(G)+M(G)\,,
    \end{align}
    for any (possibly weighted) graph $G$.
\end{claim}
\begin{proof}
    The proof is quite similar to that of \cref{claim:l_unweighted_reduction}, but we derive it here for completeness. We prove by contradiction by assuming \cref{eq:ham_reduction_hypothesis} is true and \cref{eq:ham_reduction_conclusion} is false. Let $G$ be the minimal graph for which \cref{eq:ham_reduction_conclusion} is false. Let $x_e$ denote the energy a state $\rho$ obtains on edge $e$ with respect to $H^{QMC}(G)$, (i.e. $x_e = \tr\ofb{h^{QMC}_e \rho} \; \forall e \in E$. If there is some optimal solution $\rho^*$ such that there exists an edge $e$ with $x^*_e \leq 1$, then we can find a smaller $G$ by removing $e$ as follows: let $G'$ be the graph with edge $e'$ removed from $G$. Then
    
    \begin{align*}
        \lmax(H^{QMC}\of{G'}) &\geq \sum_{e \in E \neq e'} w_e x^*_e =  \lmax(H^{QMC}\of{G}) - w_{e'} x^*_{e'} \\
        &> W(G) + M(G) - w_{e'} \geq W(G') + M(G'),
    \end{align*}
    
    where we used $\lmax(H^{QMC}\of{G}) > W(G) + M(G)$,  $-x^*_e \geq -1$, and $w_e \geq 0$ for the strict inequality and $M(G) \geq M(G')$ for the final inequality. This means $\lmax(H^{QMC}\of{G'}) > W(G') + M(G')$, so $G$ was not a minimal counterexample. Thus, we may assume that for every optimal solution $\rho^*$ of $G$, we have that $x^*_e \geq 1$ for all $e$. Now let $y^*_e \defeq x^*_e - 1$. Because $\rho^*$ is optimal, 
    \begin{align*}
        \lmax\of{H(G)} = \sum_e w_e x^*_e = \sum_e w_e (1+y^*_e) = W(G) + \sum_e w_e y^*_e.
    \end{align*} 
    Since $\sum_e w_e y^*_e > M(G)$, \cref{lem:match_k_lp} implies that some inequality among \cref{eq:lp_star_constraint,eq:lp_2vcfc_constraint,eq:lp_non-negative_constraint} must be violated by $y^*$
    However, it is shown in \cite[Lemma 1]{anshu2020} that $\sum_{j \in N(i)} y^*_{ij} \leq 1$ $\forall i \in V$, so the values $y^*$ satisfy \cref{eq:lp_star_constraint}. Because $x^*_e \ge 1$ for all $e \in E$, $y^*$ also satisfy \cref{eq:lp_non-negative_constraint}, so $y^*$ must violate an inequality among \cref{eq:lp_2vcfc_constraint} for some $S$. But then the induced subgraph $F$ indcuded by $S$ is a biconnected factor-critical graph, so \cref{eq:ham_reduction_hypothesis} says that $\lmax(H^{QMC}\of{F}) \leq W(F) + M(F) \implies \sum_{e \in G(S)} y^*_e \le M(F)$. Thus we have a contradiction.
\end{proof}

The proof for $\lmax\of{H^{EPR}(G)}$ follows identically. The proofs for the maximum and minimum eigenvalue of $XY$ ($\lmax\of{H^{XY}(G)}$ and $\lmin\of{H^{XY}(G)}$) can be easily shown to follow by analyzing $y$ defined by edge energies of 
\begin{align*}
    H = 2H^{XY}(G)-W(G),
\end{align*}
for $\lmax\of{H^{XY}(G)}$ and 
\begin{align*}
    H = -2H^{XY}(G)+W(G),
\end{align*}
for $\lmin\of{H^{XY}(G)}$ and by invoking \cref{cor:star_xy_monogamy}, which in particular implies
\begin{align*}
    \lmax\of{2H^{XY}(G)-W(G)} &\le W(G)+M(G), \\
    \lmax\of{-2H^{XY}(G)+W(G)} &\le W(G)+M(G).
\end{align*}

\section{Refuted conjectures}\label{apx:numerics}
We now discuss conjectures that we refuted, as well as additional numeric experiments that are not included in the main body

\subsection{Adjacency matrix spectrum containment}\label{apx:numerics/adj_containment}

We present a counterexample for which 
\begin{align*}
    \lmin(A(F_k(G)) \geq \lmin(A(F_{k+1}(G)),
\end{align*}

is not true, as discussed in \cref{sec:conj/token}.

\vspace{-1em}
\begin{center}
\begin{tikzpicture}[scale=0.3, every node/.style={circle, fill=cyan!30, draw=black, minimum size=0.1cm}, thick]

\node (a) at (5,12.25) {};
\node (b) at (6.25,9.75) {};
\node (c) at (9.25,9.75) {};
\node (d) at (10.75,12.25) {};
\node (e) at (7.75,12.25) {};
\node (f) at (7.75,14.75) {};
\node (g) at (2.75,12.25) {};
\node (h) at (2.75,9.75) {};

\draw (b) -- (c);
\draw (c) -- (d);
\draw (d) -- (f);
\draw (f) -- (a);
\draw (a) -- (b);
\draw (a) -- (e);
\draw (e) -- (f);
\draw (e) -- (d);
\draw (h) -- (b);
\draw (a) -- (g);
\end{tikzpicture}
\end{center}

This counterexample has $\lmin(A(F_3(G))=-4.472$ and $\lmin(A(F_4))= -4.470$.

\subsection{Upper bounds based on maximum cut}\label{apx:numerics/cut_bounds}

Based on numerical evidence for a smaller set of graphs, we additionally tested the maximum cut-based bound 

\begin{align*}
    \lmax\of{L\of{F_k\of{G}}} \le \frac{W(G)+C(G)}{2}+k,
\end{align*}

which is stronger than \cref{conj:L_max}. Among our entire suite of graphs, only one graph violated this conjecture which is depicted below

\begin{center}
\begin{tikzpicture}[scale=.8, every node/.style={circle, fill=cyan!30, draw=black, minimum size=0.3cm}, thick]
      \draw
        (0.0:2) node (0){}
        (36.0:2) node (1){}
        (72.0:2) node (2){}
        (108.0:2) node (3){}
        (144.0:2) node (4){}
        (180.0:2) node (5){}
        (216.0:2) node (6){}
        (252.0:2) node (7){}
        (288.0:2) node (8){}
        (324.0:2) node (9){};
      \begin{scope}[-]
        \draw (0) to (4);
        \draw (0) to (5);
        \draw (0) to (6);
        \draw (0) to (7);
        \draw (1) to (4);
        \draw (1) to (5);
        \draw (1) to (6);
        \draw (1) to (7);
        \draw (2) to (6);
        \draw (2) to (7);
        \draw (2) to (8);
        \draw (2) to (9);
        \draw (3) to (6);
        \draw (3) to (7);
        \draw (3) to (8);
        \draw (3) to (9);
        \draw (4) to (8);
        \draw (4) to (9);
        \draw (5) to (8);
        \draw (5) to (9);
      \end{scope}
    \end{tikzpicture}
\end{center}

This graph failed only at $k=5$, where $\lmax\of{L\of{F_k\of{G}}}=23.062$, $W(G)=20$, $C(G)=16$, and thus $23.062>\frac{W(G)+C(G)}{2}+k=23$. The graph is similar to the Peterson graph, in the sense that it can be decomposed into two $5$-cycles joined by edges. However, in the Petersen graph there is one edge connecting each vertex in the cycles to the other cycle, while in this graph there are two.

A correspondence from Wenlin Zhang and Haobo Ma refuted the following conjecture, which would have implied improved approximation ratios for XY
\begin{conjecture}\label{conj:A_cut_bound}
    For all $G$ and for all $1 \le k \le \floor{n/2}$
    \begin{align*}
        \lmin\of{A\of{F_k\of{G}}} \ge -\frac{C(G)+M_k(G)}{2}.
    \end{align*}
\end{conjecture}
The counterexample graph has 200 vertices and 100 tokens, and is shown using a violating Rayleigh vector. The counterexample is formalized in lean \cite{omega2026}.

\section{Proof of approximation ratios}\label{apx:approx}

We now present proofs of the claimed approximation ratios presented in \cref{sec:imp}. For EPR we accomplish this by carefully relating our algorithmic lower bounds and combinatorial upper bounds in terms of matchings. For QMC and XY, we also must consider the maximum cut.

\subsection{EPR}\label{apx:approx/epr}
\EPRApprox*
\begin{proof}
    Note that we can express the density matrix of $\ket{\chi}$ in \cref{eq:epr_chi} as
    \begin{align}
        \rho \defeq \ket{\chi}\bra{\chi} = \frac{1}{2^n} \bigotimes_{(i,j) \in M} \of{I + Z_i Z_j + \sin{2\theta}\of{X_iX_j -Y_iY_j} + \cos{2\theta}\of{Z_i +Z_j}} \bigotimes_{i \notin M} \of{I + Z_i},
    \end{align}
    where $(i,j)\in M$ denotes edges in $M$ and $i \notin M$ denotes vertices that are not matched by $M$. We then analytically evaluate the partial traces
    \begin{align}\label{eq:epr_energy_by_edge_type_cases}
        \tr\ofb{h^{EPR}_{ij} \rho} = 
        \begin{cases}
            \frac{1}{2} \of{2 + 2\sin{2\theta}} & \text{if } (i,j)\in M, \\
            \frac{1}{2} \of{1 + \cos{2\theta}} & \text{if } (i,j)\notin M \text{ and } \mathbf{1}\ofc{i\in M} + \mathbf{1}\ofc{j\in M} = 1, \\
            \frac{1}{2} \of{1 + \cos^2{2\theta}} & \text{if } (i,j)\notin M \text{ and } i\in M \text{ and } j\in M.
        \end{cases}
    \end{align}
    where the first equation gives the energy on an edge that is in the matching, the second gives the energy on an edge connecting a matched vertex to an unmatched vertex, and the third gives the energy of an edge connecting two vertices that are each matched to distinct vertices. As $\cos^2{2\theta} \leq \cos{2\theta}$ for all $\theta\in\ofb{0, \pi/4}$ (indeed the values of $\theta$ used hereafter will fall in this range), we can write
    \begin{align}
        \tr\ofb{h^{EPR}_{ij} \rho} &\geq \frac{1}{2} \of{2+ 2\sin{2\theta}} &&\hspace{-7em} (ij)\in M, \\
        \tr\ofb{h^{EPR}_{ij} \rho} &\geq  \frac{1}{2}\of{1+ \cos^2{2\theta}} &&\hspace{-7em} (ij) \notin M.
    \end{align}
    Thus we have 
    \begin{align}
        \tr{H^{EPR}(G) \rho} &\geq \frac{1+\cos^2{2\theta}}{2} W + \frac{2+2\sin{2\theta}-(1+\cos^2{2\theta})}{2} M \\
        &= \frac{1+\cos^2{2\theta}}{2} W + \frac{1+2\sin{2\theta}-\cos^2{2\theta}}{2} M \\
        & \defeq \frac{1+(1-\gamma^2)}{2} W + \frac{1+2\gamma-(1-\gamma^2)}{2} M \\
        & \defeq c_W(\gamma) W + c_M(\gamma) M,
    \end{align}

    where in the second-to-last line we defined $\gamma \defeq \sin{2\theta}$ and the last line we introduce $c_W(\gamma)$ and $c_M(\gamma)$ for the coefficients in front of $W$ and $M$, respectively. Now, we choose $\gamma$ to be $\gamma^*$ such that

    \begin{align}\label{eq:sol_c_m_c_w_ratio_is_a}
        &c_M(\gamma^*)=c_W(\gamma^*)=\nu, \\
        \implies &\gamma^* = \frac{-1\pm\sqrt{5}}{2}.
    \end{align}

    Then, the positive root yields $\nu = \frac{1+\sqrt{5}}{4}$. Thus,  we obtain the approximation ratio 
    \begin{align}
        \alpha &\ge \frac{c_W(\gamma^*) \,W + c_M(\gamma^*)\, M}{\lmax\of{H^{EPR}(G)}} \ge \nu \cdot \frac{W+M}{W+M} = \nu \approx 0.809,
    \end{align}
    where in the second inequality we invoke \cref{lem:h_bounds}.
\end{proof}

\subsection{XY}\label{apx:approx/xy}
For the approximation ratios of XY, we first use the following observation 

\begin{lemma}\label{lem:xy_cut_bound}
    For all graphs $G$, the following bound holds
    \begin{align*}
        \lmax\of{H^{XY}(G)} \le 2 \,C(G) - \frac{W(G)}{2}\,.
    \end{align*}
\end{lemma}
\begin{proof}
First note that the Hamiltonian can be decomposed as follows
\begin{align*}
    H^{XY}(G) &=  \frac{1}{2}\sum_{(i,j)\in E} \of{I_i I_j - X_i X_j} +\frac{1}{2}\sum_{(i,j)\in E} \of{I_i I_j - Y_i Y_j}  -  \frac{1}{2}\sum_{(i,j)\in E} I_i I_j ,
\end{align*}
Now, note that 
\begin{align*}
    \frac{1}{2} \sum_{(i,j)\in E} I_i I_j - Z_i Z_j,
\end{align*}
encodes the Hamiltonian describing the classical MaxCut problem in the computational basis. Replacing the Pauli $Z$ with $X$ or $Y$ encodes MaxCut in the $X$ or $Y$ basis, respectfully. Thus, by the triangle inequality, we have 
\begin{align*}
    \lmax\of{H^{XY}(G)} \le 2 \,C(G) - \frac{W(G)}{2}.
\end{align*}
\end{proof}

\XYApprox*
\begin{proof}
    We can first analytically evaluate the partial traces of $\rho_C$
    \begin{align}\label{eq:xy_rho_c_energy}
        \tr\ofb{h^{XY}_{ij} \rho_C} =
        \begin{cases}
            1 & \text{if } (i,j)\in C(G), \\
            0 & \text{if } (i,j)\notin C(G).
        \end{cases}
    \end{align}
    We then do the same for $\rho_M$
    \begin{align*}
        \tr\ofb{h^{XY}_{ij} \rho_M} =
        \begin{cases}
            \frac{3}{2} & \text{if } (i,j)\in M(G), \\
            \frac{1}{2} & \text{if } (i,j)\notin M(G).
        \end{cases}
    \end{align*}
    Thus we see
    \begin{align*}
        \tr\ofb{H^{XY}(G)\rho_C} &= C(G), \\
        \tr\ofb{H^{XY}(G)\rho_M} &= \frac{3}{2}M(G) + \frac{1}{2}(W(G)-M(G)) = M(G)+\frac{W(G)}{2}. \\
    \end{align*}
    We then use the upper bounds from \cref{lem:h_bounds} and \cref{lem:xy_cut_bound} to show
    \begin{align*}
        \alpha = \min_G\frac{\max\of{\tr\ofb{H^{XY}(G) \rho_C},\tr\ofb{H^{XY}(G) \rho_M}}}{\lmax\of{H^{XY}(G)}} \ge \frac{\max\of{C(G),M(G)+\frac{W(G)}{2}}}{\min\of{2C(G)-\frac{W(G)}{2}, W(G)+\frac{M(G)}{2}}}\,.
    \end{align*}
    We can minimize this quantity over $c\defeq C(G)/W(G)$, $M \defeq M(G)/W(G)$ to find
    \begin{align}\label{eq:xy_approx_minimization}
        \alpha \ge \min_{\substack{0 \le m \le 1\\ \frac{m+1}{2}\le c \le 1}} \frac{\max\of{c,m+\frac{1}{2}}}{\min\of{2c-\frac{1}{2}, 1+\frac{m}{2}}}.
    \end{align}
    This yields $\alpha \ge 5/7\approx0.714$ at $c\rightarrow 5/6$, $m\rightarrow 1/3$.
\end{proof}

\XYApproxEfficient*
\begin{proof}
    The proof is identical to that of \cref{cor:xy_approx}, albeit we take an expectation over the randomized product states $\rho_{C'}$ in \cref{eq:xy_rho_c_energy}. Following the rest of the steps, may simply evaluate \cref{eq:xy_approx_minimization}, but replace the $c$ in the numerator by $0.9349 c$. The minimization then yields $\alpha \ge 0.674$ at  $c \rightarrow 0.816$, $m \rightarrow 0.263$.
\end{proof}

\subsection{QMC}\label{apx:approx/qmc}
For QMC we use the upper bound
\begin{lemma}[Equation 18 \cite{anshu2020}]\label{lem:qmc_cut_bound}
    For all graphs $G$, the following bound holds
    \begin{align*}
        \lmax\of{H^{QMC}(G)} \le 3 \,C(G) - W(G).
    \end{align*}
\end{lemma}
\begin{proof}
    Our proof is essentially the same as \cite{anshu2020}, albeit with different notations and definitions. We write it here for completeness. First note that the Hamiltonian can be decomposed as follows
    \begin{align*}
        H^{QMC}(G) &= \frac{1}{2} \sum_{(i,j)\in E} \of{I_i I_j - X_i X_j} +\frac{1}{2}\sum_{(i,j)\in E} \of{I_i I_j - Y_i Y_j} +\frac{1}{2}\sum_{(i,j)\in E} \of{I_i I_j - Z_i Z_j}  -  \sum_{(i,j)\in E} I_i I_j ,
    \end{align*}
    
    Again the sum over terms $X_iX_j$, $Y_iY_j$, $Z_iZ_j$ each encode the Hamiltonian describing the classical MaxCut problem in the three Pauli bases. Thus, by the triangle inequality, we have 
    
    \begin{align*}
        \lmax\of{H^{QMC}(G)} \le 3 \,C(G) - W(G).
    \end{align*}
\end{proof}

\QMCMATCHApprox*
\begin{proof}
    The proof proceeds similarly to that for XY. We first can analytically evaluate the partial traces of $\rho_C$
    \begin{align}\label{eq:qmc_rho_c_energy}
        \tr\ofb{h^{QMC}_{ij} \rho_C} =
        \begin{cases}
            1 & \text{if } (i,j)\in C(G), \\
            0 & \text{if } (i,j)\notin C(G).
        \end{cases}
    \end{align}
    We then do the same for $\rho_M$
    \begin{align}\label{eq:qmc_rho_m_energy}
        \tr\ofb{h^{QMC}_{ij} \rho_M} =
        \begin{cases}
            2 & \text{if } (i,j)\in M(G), \\
            \frac{1}{2} & \text{if } (i,j)\notin M(G).
        \end{cases}
    \end{align}
    Thus we see
    \begin{align*}
        \tr\ofb{H^{QMC}(G)\rho_C} &= C(G), \\
        \\tr\ofb{H^{QMC}(G)\rho_M} &= 2M(G) + \frac{1}{2}(W(G)-M(G)) = \frac{3M(G)+W(G)}{2}. \\
    \end{align*}
    We then use the upper bounds from \cref{lem:h_bounds} and \cref{lem:qmc_cut_bound} to show
    \begin{align*}
        \alpha = \min_G\frac{\max\of{\tr\ofb{H^{QMC}(G) \rho_C},\tr\ofb{H^{QMC}(G) \rho_M}}}{\lmax\of{H^{QMC}(G)}} \ge \frac{\max\of{C(G),\frac{3M(G)+W(G)}{2}}}{\min\of{3C(G)-W(G), \frac{W(G)+M(G)}{2}}}\,.
    \end{align*}
    Again using $c\defeq C(G)/W(G)$, $M \defeq M(G)/W(G)$ we find
    \begin{align}\label{eq:qmc_approx_minimization}
        \alpha \ge \min_{\substack{0 \le m \le 1\\ \frac{m+1}{2}\le c \le 1}} \frac{\max\of{c,\frac{3m+1}{2}}}{\min\of{3c-1, 1+m}}\,.
    \end{align}
    This yields $\alpha \ge 5/8=0.625$ at $c\rightarrow 5/7$, $m\rightarrow 1/7$.
\end{proof}

\QMCMATCHApproxEfficient*
\begin{proof}
    Again using \cite{briet2010}, we have that for QMC, we can prepare a randomized product state with expected energy $\nu(3)=0.9563$ of $C$. We again follow the proof of \cref{cor:qmc_match_approx}, only replacing \cref{eq:qmc_rho_c_energy} with an expectation value over the randomized states $\rho_{C'}$. Following the rest of the steps, we arrive at \cref{eq:qmc_approx_minimization} but with $0.9563 \,c$ replacing $c$. This yields $\alpha \ge 00.604$ at $c\rightarrow 0.705$ and $m\rightarrow .116$.
\end{proof}

\QMCPMATCHApprox*
\begin{proof}
    The result is stated at the end of \cite[Section 4.2]{apte2025}. It is found by numerically minimizing the function in \cite[Equation 19]{apte2025} with the value $d=1$ in \cite[Lemma 3]{apte2025}. The value $d=1$ is satisfied by explicitly ensuring matching constraints are met via \cref{alg:sep_oracle}.
\end{proof}

\section{Ear decompositions}\label{apx:ear_decomp}
In \cref{sec:conj/ham/reduction}, we showed that for our conjectured bounds on the maximum energy of QMC, XY, and EPR, it suffices to consider biconnected factor-critical graphs. Factor-critical graphs are often defined in terms of their \emph{odd ear decompositions} (defined below in \cref{def:odd_ear_decomp}). This provides an intuitive framework for proving the conjectures via induction. In this section, we show that the constraint of biconnectivity actually implies that the graphs have \emph{odd open ear decompositions}. We then prove the conjectures hold for the base case of odd-order cycle graphs. In order to prove the conjectures then, it suffices to complete the inductive step, which we have been unable to do. We first define 

\begin{definition}[Odd ear decomposition]\label{def:odd_ear_decomp}
    Let an ear $Y(V,E)$ be an unweighted path or cycle graph with an odd number of edges. Then, given an unweighted graph $G(V, E)$, an ear decomposition $G$ of a graph is an edge-disjoint partition of $G$ into a sequence of $\ell$ subgraphs, which we refer to as \emph{ears} 
    \begin{align*}
        \of{Y_i(V_i, E_i)}_{i=1}^{\ell}= \Big(Y_1(V_1,E_1), Y_2(V_2,E_2), \ldots, Y_{\ell}(V_{\ell},E_{\ell}) \Big),
    \end{align*}
    with the property that for each $i \in \ofc{2,3,\ldots,\ell}$ we have either
    \begin{enumerate}[label=\alph*)]
        \item if $Y_i$ is a path, the two endpoints of the path are in $V_{T_i} \defeq \cup_{j \in [i-1]} V_j$, and none of the non-endpoint vertices are in $V_{T_i}$ (denoted by an \emph{open ear}),
        \item if $Y_i$ is a cycle, exactly one vertex of $Y_i$ is in $V_{T_i}$ (denoted by a \emph{closed ear}),
    \end{enumerate}
    and $Y_1$ is a cycle (closed) ear. 
\end{definition}

\begin{definition}[Odd open-ear decomposition]
    An odd open-ear decomposition of a graph $G$ is an odd-ear decomposition such that all but the first ear are open ears.
\end{definition}

We may then borrow the following lemma
\begin{lemma}[\cite{lovasz1972}]\label{lem:fc_ear}
    Any factor-critical graph $G$ admits an odd-ear decomposition.
\end{lemma}

The constraint of biconnectivity implies the stronger decomposition

\begin{restatable}{lemma}{BiconnectedOpenEar}\label{lem:biconnected_open_ear}
    Any graph $G \in \mathcal{F}$ admits an odd open-ear decomposition.
\end{restatable}
\begin{proof}
    As $G$ is factor-critical, we know from \cref{lem:fc_ear} that $G$ admits an odd-ear decomposition. Let us write the sequence of ears as $Y=(C_1,O_2,C_3,C_4,\ldots)$ where $C_i$ is a closed ear and $O_i$ is an open ear. Let $C_s$ be the closed ear that appears last in the sequence, let $F$ be the subgraph formed by all the ears appearing before, but not including, $C_s$, and let $u'$ be the single vertex that joins $C_s$ with $F$. Note that $C_s$ cannot be the last ear in the sequence, or else removing $u'$ disconnects the graph. Furthermore, one of the open ears that follows $C_s$ must have one endpoint in $F$, or else removing $u'$ again disconnects the graph. Let $O_f$ be the first such open ear and let the vertex that connects $O_f$ to $F$ be called $v'$. We then focus on the subgraph $F'$ consisting of edges in the subsequence $Y'=(C_s,O_{s+1},O_{s+2},\ldots,O_f)$. We claim that we may find an alternate decomposition $T'$ of $F'$ which only consists of open ears. If this is true, then we may remove the last closed ear and proceed by induction on the rest of the sequence until we reach $C_1$, completing our proof. We now prove this claim. We form $T'$ by construction: 
    \begin{enumerate}
        \item Initialize $T'$ with a single open ear from $v'$ to $u'$ as follows. First traverse all edges in the final open ear $O_f$. This path starts from $v'$, ends at some vertex $t_0$, and has an odd number of edges. Now, the vertex $t_0$ must be in some ear of $Y'$. Let us label this ear $O_{\beta_1}$. As $O_{\beta_1}$ has an odd number of edges, we may always next choose the path from $t_0$ to one of the endpoints $t_1$ of $O_{\beta_1}$ such that the number of edges in this path is \emph{even}. Now, $t_1$ is in some other ear $O_{\beta_2}$ that appears previously in the sequence $Y'$. As such, we may continue this process until we reach a vertex $t_{f}$ in $C_s$. Then, complete the open ear with the path from $t_{f}$ to $v'$ that is of \emph{even} length.
        \item Append to $T'$ the open ear that contains all the edges in $C_s$ not covered by the previous step.
        \item Proceed through the all the ears $Y_i \in Y':\; Y_i \neq C_s, Y_i\neq O_f$ in the order of $Y'$. Append to $T'$ either the entire ear $Y_i$ if $Y_i$ is uncovered by any of the edges in the previous step, or append to $T'$ the remaining edges edges of $Y_i$ uncovered by ears in the previous step.
    \end{enumerate}
    This algorithm only consists of open ears. It can easily be seen to be efficient. Thus is remains to show that this algorithm terminates as expected and induces a valid odd open-ear decomposition. We first show termination of step (1). We may order the vertices in $V'$ in $F'$ in the order that they appear in the sequence $Y'$. Denote this ordering by the sequence $L$. For instance, the first vertex in $L$ is $u'$ and the last is $v'$. Now note that each step $i$ of (1) ends at some $t_i$. By construction, this vertex appears in $L$ before any of the other vertices $j$ in the current open ear. As this is true for each step, we must eventually reach some vertex in $C_s$. The remaining steps clearly terminate. To show that this algorithm forms a valid odd-ear decomposition, first note that (1) by construction forms an open ear, as it is an odd length path from vertices $u'$ to $v'$ that are both in $F'$. Then, (2) forms an open ear, as (1) covers an even number of edges in $C_s$, so the remaining number of edges is odd. Furthermore, the endpoints of the path are contained in the union of $F'$ and the odd ear from (1). Finally, all steps of (3) form odd ears, as (1) covers an even number of edges from each open ear that it traverses, thus the number of edges that each ear in (3) traverses is odd. Furthermore, the endpoints of each ear are in the subgraph formed by the union of $F'$ and previous ears as either both endpoints are covered by $Y'$ by definition or one of the endpoints was covered by the odd ear in (1). 
\end{proof}

We now prove the base case of our conjectures by showing that they hold on the first closed ear

\begin{restatable}{lemma}{CycleBound}\label{lem:cycle_bound}
    For any unweighted cycle graph $C_n$, $n \ge 2$, the hypotheses \cref{eq:qmc_reduction}, \cref{eq:epr_reduction} and  \cref{eq:xy_max_reduction}, \cref{eq:xy_min_reduction} hold.
\end{restatable}

We first provide the following table, which we compute via exact diagonalization on a computer 
\cite{apte2025a}

\begin{table}[H] 
    \resizebox{.95\textwidth}{!}{
    \begin{tabular}{|c||c|c|c|c|c|c|c|c|c|c|c|c|}
        \hline
        $n$        & 2   & 3   & 4    & 5    & 6    & 7    & 8     & 9     & 10    & 11    & 12    & 13    \\
        \hline
        \rule{0pt}{3ex}   
        $\lmax\of{H^{QMC}\of{C_n}}$ & 2.00 & 3.00 & 6.00 & 6.24 & 8.61 & 9.21 & 11.30 & 12.09 & 14.03 & 14.94 & 16.77 & 17.76 \\
        $\lmax\of{H^{EPR}\of{C_n}}$ & 2.00 & 4.00 & 6.00 & 6.83 & 8.61 & 9.63 & 11.30 & 12.42 & 14.03 & 15.20 & 16.77 & 17.98 \\
        $\lmax\of{H^{EPR}\of{P_n}}$ & 2.00 & 3.00 & 4.73 & 5.86 & 7.49 & 8.67 & 10.25 & 11.47 & 13.02 & 14.26 & 15.78 & 17.05 \\
        $\lmax\of{H^{XY}\of{C_n}}$  & 0.50 & 1.50 & 4.00 & 4.12 & 6.46 & 6.75 & 8.83 & 9.26 & 11.16 & 11.70 & 13.46 & 14.09 \\
        $\lmax\of{-H^{XY}\of{C_n}}$ & -0.50 & -1.50 & 0.00 & -0.50 & 0.46 & 0.10 & 0.83 & 0.56 & 1.16 & 0.96 & 1.46 & 1.31 \\
        $\lmax\of{H^{XY}\of{P_n}}$  & 0.50 & 1.00 & 3.12 & 3.73 & 5.55 & 6.26 & 7.91 & 8.70 & 10.24 & 11.08 & 12.56 & 13.43 \\
        \hline
    \end{tabular}
    }
    \caption{Extremal energies for QMC, EPR, and XY Hamiltonians on path and ring graphs up to $13$ nodes. We do not show $\lmax\of{H^{EPR}\of{P_n}}$ and $-\lmin\of{H^{XY}\of{P_n}}$, as $\lmax\of{H^{EPR}\of{P_n}}=\lmax\of{H^{QMC}_{P_n}}$ and $\lmax\of{H^{XY}\of{P_n}-\frac{W}{2}I}=\lmax\of{-H^{XY}\of{P_n}+\frac{W}{2}I}$  due to \cref{lem:bipartite_equivalence}.}
    \label{tab:path_cycle_opts}
\end{table}

We first show the following useful lemmas for EPR on path graphs (and thus for QMC) \cref{lem:bipartite_equivalence})
\begin{lemma}\label{lem:path_match_bound_with_a}
    For an unweighted path graph $P_{2k}$ with $k\in \mathbb{Z}^+$ and $a \in \mathbb{R}$
    \begin{align}
        &\lmax\of{H^{EPR}(P_{2k})} \leq W\of{P_{2k}}+M\of{P_{2k}}-a \\
        &\implies \lmax\of{H^{EPR}(P_{2(k+1)})} \leq W\of{P_{2(k+1)}}+M\of{P_{2(k+1)}}-a. 
    \end{align}
    In particular, this fact, combined with \cref{tab:path_cycle_opts}, shows that for any $k \geq 6$, $\lmax\of{P_{2k}} \leq W\of{P_{2k}}+M\of{P_{2k}}-1 = 3k-2$.
\end{lemma}
\begin{proof}
    We have 
    \begin{align}
        \lmax\of{H^{EPR}(P_{2(k+1)})} 
        &\leq \lmax\of{H^{EPR}(P_{2k})} + \lmax\of{H^{EPR}(P_3)} \\
        &\leq W\of{P_{2k}}+M\of{P_{2k}}-a+3 \\
        &= W\of{P_{2(k+1)}}+M\of{P_{2(k+1)}} - a, 
    \end{align}
    where in the first line we used the triangle inequality, in the second line we used the hypothesis in the lemma and $\lmax\of{P_3}$ from \cref{tab:path_cycle_opts}, and in the last line we used that $W\of{P_{2(k+1)}} = W\of{P_{2k}}+2$ and $M\of{P_{2(k+1)}} = M\of{P_{2k}}+1$.
\end{proof}

We now may prove \cref{lem:cycle_bound} for EPR. The result follows for QMC by \cref{cor:qmc_less_than_epr}
\CycleBound*
\begin{lemma}\label{lem:odd_cycle_match_bound}
    For an unweighted odd-length cycle graph $C_{2k+1}$ with $k\in \mathbb{Z}^+$ and $a \in \mathbb{R}$
    \begin{align}
        \lmax\of{H^{EPR}(C_{2k+1})} \leq W\of{C_{2k+1}}+M\of{C_{2k+1}}.
    \end{align}
\end{lemma}
\begin{proof}
    For $k \leq 6$, the Lemma can be verified for both QMC and EPR numerically, as seen in \cref{tab:path_cycle_opts}. For $k \geq 6$, we have
    \begin{align}
        \lmax\of{H^{EPR}(C_{2k+1})} 
        &\leq \lmax\of{H^{EPR}(P_{2k})} + \lmax\of{H^{EPR}(P_3)} \\
        &\leq W\of{H^{EPR}(P_{2k})}+M\of{H^{EPR}(P_{2k})}-1+3 \\
        &= W\of{C_{2k+1}}+M\of{P_{2k+1}}, 
    \end{align}
    where in the first line we used the triangle inequality, in the second line we used \cref{lem:path_match_bound_with_a} and $\lmax\of{P_3}$ from \cref{tab:path_cycle_opts}, and in the last line we used that $W\of{C_{2k+1}} = W\of{P_{2k}}+2$ and $M\of{C_{2k+1}} = M\of{P_{2k}}$.
\end{proof}

The base case for the minimum and maximum of XY in \cref{lem:cycle_bound} follow via similar arguments. 

\end{document}